\def\confversion{0}
\def\ifconf{\ifnum\confversion=1}
\def\ifnotconf{\ifnum\confversion=0}

\documentclass[11 pt]{article}

\def\showauthornotes{1}
\def\showkeys{0}
\def\showdraftbox{0}

\usepackage{xspace,xcolor,enumerate}
\usepackage{amsmath,amssymb}
\usepackage{amsthm}
\usepackage[toc,page]{appendix}
\usepackage{thmtools}
\usepackage{thm-restate}
\usepackage{color,graphicx}

\ifnum\showkeys=1
\usepackage[color]{showkeys}
\fi

\definecolor{darkred}{rgb}{0.5,0,0}
\definecolor{darkgreen}{rgb}{0,0.35,0}
\definecolor{darkblue}{rgb}{0,0,0.55}

\usepackage[pdfstartview=FitH,pdfpagemode=UseNone,colorlinks,linkcolor=darkblue,filecolor=darkred,citecolor=darkgreen,urlcolor=darkred,pagebackref]{hyperref}

\usepackage[capitalise,nameinlink]{cleveref}
\usepackage[T1]{fontenc}
\usepackage{mathtools,dsfont,bbm}
\usepackage{mathpazo}

\ifnum\showauthornotes=1
\newcommand{\Authornote}[2]{{\sf\small\color{red}{[#1: #2]}}}
\newcommand{\Authorcomment}[2]{{\sf \small\color{gray}{[#1: #2]}}}
\newcommand{\Authorfnote}[2]{\footnote{\color{red}{#1: #2}}}
\else
\newcommand{\Authornote}[2]{}
\newcommand{\Authorcomment}[2]{}
\newcommand{\Authorfnote}[2]{}
\fi

\ifnum\showdraftbox=1
\newcommand{\draftbox}{\begin{center}
  \fbox{%
    \begin{minipage}{2in}%
      \begin{center}%
        \begin{Large}%
          \textsc{Working Draft}%
        \end{Large}\\
        Please do not distribute%
      \end{center}%
    \end{minipage}%
  }%
\end{center}
\vspace{0.2cm}}
\else
\newcommand{\draftbox}{}
\fi

\newtheorem{theorem}{Theorem}[section]

\newtheorem{observation}[theorem]{Observation}
\newtheorem{definition}[theorem]{Definition}
\newtheorem{lemma}[theorem]{Lemma}

\newtheorem{proposition}[theorem]{Proposition}
\newtheorem{corollary}[theorem]{Corollary}
\newtheorem{claim}[theorem]{Claim}
\newtheorem{fact}[theorem]{Fact}

\def\FullBox{\hbox{\vrule width 6pt height 6pt depth 0pt}}

\def\qed{\ifmmode\qquad\FullBox\else{\unskip\nobreak\hfil
\penalty50\hskip1em\null\nobreak\hfil\FullBox
\parfillskip=0pt\finalhyphendemerits=0\endgraf}\fi}

\def\qedsketch{\ifmmode\Box\else{\unskip\nobreak\hfil
\penalty50\hskip1em\null\nobreak\hfil$\Box$
\parfillskip=0pt\finalhyphendemerits=0\endgraf}\fi}

\newenvironment{proofsketch}{\begin{trivlist} \item {\bf
Proof Sketch:~~}}
  {\qedsketch\end{trivlist}}

\newenvironment{proofof}[1]{\begin{trivlist} \item {\bf Proof of
#1:~~}}
  {\qed\end{trivlist}}

\def\to{\rightarrow}
\def\eps{\varepsilon}
\def\epsilon{\varepsilon}

\def\eps{\epsilon}

\def\phi{\varphi}
\def\cal{\mathcal}

\newcommand{\defeq}{:=}

\newcommand{\ie}{i.e.,\xspace}

\newcommand{\etal}{et al.\xspace}

\newcommand{\mper}{\,.}
\newcommand{\mcom}{\,,}

\newcommand{\R}{{\mathbb R}}
\newcommand{\E}{{\mathbb E}}

\newcommand{\N}{{\mathbb{N}}}

\newcommand{\pmone}{\{-1,1\}\xspace}
\newcommand{\on}{\pmone}

\newcommand{\indicator}[1]{\mathds{1}_{\{#1\}}}

\newcommand{\gaussian}[2]{{\cal N(#1, #2)}}

\usepackage{nicefrac}

\let\nfrac=\nicefrac

\newcommand{\abs}[1]{\ensuremath{\left\lvert #1 \right\rvert}}

\makeatletter
\def\norm#1{
  \@ifnextchar\bgroup%
   {\normalpnorm{#1}}%
   {\defaultnorm{#1}}%
}
\def\defaultnorm#1{%
    \renderNorm{#1}{}
}
\def\normalpnorm#1#2{%
   \@ifnextchar\bgroup%
   {\banaxnorm{#1}{#2}}%
   {\renderNorm{#2}{#1}}
}

\def\banaxnorm#1#2#3{%
    \renderNorm{#3}{#1\to#2 }
}

\def\renderNorm#1#2{%
    \@ifnextchar^%
    {\fixExponent{#1}{#2}}%
    {\ensuremath{\mathchoice%
        {\lVert #1 \rVert_{#2}}%
        {\lVert #1 \rVert_{#2}}%
        {\lVert #1 \rVert_{#2}}%
        {\lVert #1 \rVert_{#2}}}%
    }%
}
\def\fixExponent#1#2^#3{%
    \ensuremath{{\lVert #1 \rVert^{#3}_{#2}}}%
}

\def\enorm#1#2{%
   \@ifnextchar\bgroup%
   {\norm{L_{#1}}{L_{#2}}}%
   {\norm{L_{#1}}{#2}}%
}

\def\cnorm#1#2{%
   \@ifnextchar\bgroup%
   {\norm{\ell_{#1}}{\ell_{#2}}}%
   {\norm{\ell_{#1}}{#2}}%
}

\makeatother

\newcommand{\mysmalldot}[2]{\ensuremath{\langle #1, #2 \rangle}}
\newcommand{\ip}[1]{\left\langle #1 \right\rangle}

\newcommand{\cip}[1]{\left\langle #1 \right\rangle}  
\newcommand{\eip}[1]{\left\langle #1 \right\rangle} 

\newcommand{\Esymb}{\mathbb{E}}
\newcommand{\Psymb}{\mathbb{P}}

\makeatletter
\def\Pr#1{%
    \ProbabilityRender{\Psymb}{#1}%
}

\def\Ex#1{%
    \ProbabilityRender{\Esymb}{#1}%
}

\def\ProbabilityRender#1#2{
  \@ifnextchar\bgroup%
  {\renderwithdist{#1}{#2}}
   {\singlervrender{#1}{#2}}
}
\def\singlervrender#1#2{%
   \ensuremath{\mathchoice
       {{#1}\left[ #2 \right]}
       {{#1}[ #2 ]}
       {{#1}[ #2 ]}
       {{#1}[ #2 ]}
   }
}
\def\renderwithdist#1#2#3{%
   \@ifnextchar\bgroup
   {\superfancyrender{#1}{#2}{#3}}
   {\ensuremath{\mathchoice
      {\underset{#2}{#1}\left[ #3 \right]}
      {{#1}_{#2}[ #3 ]}
      {{#1}_{#2}[ #3 ]}
      {{#1}_{#2}[ #3 ]}
     }
   }
}
\def\superfancyrender#1#2#3#4#5{
   \ensuremath{\mathchoice
      {\underset{#1}{{#1}}\left#4 #3 \right#5}
      {{#1}_{#2}#4 #3 #5}
      {{#1}_{#2}#4 #3 #5}
      {{#1}_{#2}#4 #3 #5}
   }
}
\makeatother

\newcommand\restrict[1]{\raisebox{-.5ex}{$|$}_{#1}}

\newfont{\inhead}{eufm10 scaled\magstep1}

\newcommand{\calD}{{\cal D}}

\newcommand{\calL}{{\cal L}}

\newcommand{\poly}{{\mathrm{poly}}}

\newcommand{\ee}{\ensuremath{\mathrm e}}

\newcommand{\RRR}{\ensuremath{\mathbb{R}}}

\newcommand{\classfont}[1]{\textsf{#1}}
 
\newcommand{\BPP}{\classfont{BPP}\xspace}

\newcommand{\NP}{\classfont{NP}\xspace}

\newcommand{\DTIME}[1]{\classfont{DTIME}\inparen{#1}\xspace}
\newcommand{\BPTIME}[1]{\classfont{BPTIME}\inparen{#1}\xspace}

\newcommand{\opt}{{\sf OPT}\xspace}

\newcommand{\problemmacro}[1]{\textsf{#1}}

\newcommand{\labelcover}{\problemmacro{Label Cover}\xspace}                                

\newcommand{\threesat}{3-SAT\xspace}
\newcommand{\groth}[2]{\ensuremath{(#1, #2)}-Grothendieck}
\newcommand{\onorm}[2]{\ensuremath{#1{\to}#2} norm}

\newcommand{\inparen}[1]{\left(#1\right)}             
\newcommand{\inbraces}[1]{\left\{#1\right\}}           

\newcommand{\Holder}{H\"{o}lder}

\newcommand{\Odonnell}{O'Donnell}
\newcommand{\Briet}{Bri\"et\xspace}          
\newcommand{\Esseen}{Ess\'een\xspace}          

\newcommand{\barz}{\overline{z}}

\newcommand{\hatf}{\widehat{f}}
\newcommand{\hatg}{\widehat{g}}

\newcommand{\bff}{{\mathbf f}}
\newcommand{\bfgg}{{\mathbf g}}

\newcommand{\bfF}{{\mathbf F}}
\newcommand{\bfA}{{\mathbf A}}

\newcommand{\bfhatf}{{\bf \widehat{f}}}
\newcommand{\bfhatg}{{\bf \widehat{g}}}
\newcommand{\bfhatP}{{\bf \widehat{P}}}
\newcommand{\bfhatL}{{\bf \widehat{L}}}

\newcommand{\elnote}{\Authornote{EL}}

\usepackage[top=1in, bottom=1in, left=1.25in, right=1.25in]{geometry}

\begin{document}


\title{Inapproximability of Matrix $p\rightarrow q$ Norms}
\author{
$\qquad$
Vijay Bhattiprolu\thanks{Supported by NSF CCF-1422045 and CCF-1526092.  {\tt
vpb@cs.cmu.edu}. Part of the work was done while visiting NTU, Singapore.} \and Mrinalkanti Ghosh\thanks{Supported by NSF CCF-1254044 \tt
mkghosh@ttic.edu} \and Venkatesan Guruswami\thanks{Supported in part by NSF
grant CCF-1526092. {\tt guruswami@cmu.edu}. Part of the work was done while visiting NTU, Singapore. }$\qquad$ \and Euiwoong
Lee\thanks{Supported by the Simons Institute for the Theory of Computing.
{\tt euiwoong@cims.nyu.edu}} \and Madhur Tulsiani \thanks{Supported by NSF
CCF-1254044 \tt madhurt@ttic.edu} 
}


\setcounter{page}{0}
\date{}

\maketitle
\draftbox
\thispagestyle{empty}


We study the problem of computing the \onorm{p}{q} of a matrix $A \in \R^{m \times n}$, defined as
\[
\norm{p}{q}{A} ~\defeq~ \max_{x \in \R^n \setminus \{0\}} \frac{\norm{q}{Ax}}{\norm{p}{x}} \mper 
\]
This problem generalizes the spectral norm of a matrix ($p=q=2$) and the Grothendieck problem
($p=\infty$, $q=1$), and has been widely studied in various regimes. 
When $p \geq q$, the problem exhibits a  dichotomy: constant factor approximation algorithms are
known if $2 \in [q,p]$, and the problem is hard to approximate within almost polynomial  factors
when $2 \notin [q,p]$.

The regime when $p < q$, known as \emph{hypercontractive norms}, is particularly significant for
various applications but much less well understood.  The case with $p = 2$ and $q > 2$ was studied
by [Barak \etal, STOC'12] who gave sub-exponential algorithms for a promise version of the problem
(which captures small-set expansion) and also proved hardness of approximation results based on the
Exponential Time Hypothesis. However, no NP-hardness of approximation is known for
these problems for any $p < q$. 

We prove the first NP-hardness result for approximating hypercontractive norms. 
We show that for any $1< p < q < \infty$ with $2 \notin [p,q]$,
  $\norm{p}{q}{A}$ is hard to approximate within $2^{O((\log
  n)^{1-\eps})}$ assuming $\NP \not \subseteq \BPTIME{2^{(\log
    n)^{O(1)}}}$. 


En route to the above result, we also prove new results for the case when $p \geq q$ with $2 \in [q,p]$.
For such $p$ and $q$, we show that
  $\norm{p}{q}{A}$ is hard to approximate within any factor smaller
  than $1/(\gamma_{p^\ast} \cdot \gamma_q)$, where for any $r$, $\gamma_r$
  denotes the $r^{th}$ norm of a standard normal random variable, and $p^\ast := p/(p-1)$ is the dual norm of $p$. The hardness factor is tight for the cases when $p$ or $q$ equals $2$.
%


\newpage

\ifnotconf
\pagenumbering{roman}
\tableofcontents
\clearpage
\fi

\pagenumbering{arabic}
\setcounter{page}{1}




\section{Introduction}\label{sec:intro}
We consider the problem of finding the \onorm{p}{q} of a given matrix $A \in \R^{m \times n}$,
which is defined as 
\[
\norm{p}{q}{A} ~\defeq~ \max_{x \in \R^n \setminus \{0\}} \frac{\norm{q}{Ax}}{\norm{p}{x}} \mper 
\]
The quantity $\norm{p}{q}{A}$ is a natural generalization of the well-studied spectral norm, which
corresponds to the case $p=q=2$. For general $p$ and $q$, this quantity computes the maximum
distortion (stretch) of the operator $A$ from the normed space $\ell_p^n$ to $\ell_q^m$.

The case when $p = \infty$ and $q = 1$ is the well known Grothendieck problem \cite{KN12, Pisier12},
where the goal is to maximize $\ip{y, Ax}$ subject to $\norm{\infty}{x}, \norm{\infty}{y} \leq 1$. In
fact, via simple duality arguments (see \cref{sec:prelims}), the general problem computing
$\norm{p}{q}{A}$ can be seen to be equivalent to the following variant of the Grothendieck problem
(and to $ \norm{q^*}{p^*}{A^T} $)
\[
\norm{p}{q}{A} ~=~ 
\max_{ \substack{\norm{p}{x} \leq 1 \\ \norm{q^*}{y} \leq 1}} \ip{y, Ax} 
~=~ \norm{q^*}{p^*}{A^T} 
\mcom
\]
where $p^*, q^*$ denote the dual norms of $p$ and $q$, satisfying $1/p + 1/p^* = 
1/q + 1/q^* = 1$.

\paragraph{Hypercontractive norms.}
The case when $p < q$, known as the case of \emph{hypercontractive} norms, also has a special
significance to the analysis of random walks, expansion and related problems in hardness of
approximation \cite{Biswal11, BBHKSZ12}. 
The problem of computing $\norm{2}{4}{A}$ is also known to be equivalent to
determining the maximum acceptance probability of a quantum protocol with multiple unentangled
provers, and is related to several problems in quantum  information theory \cite{HM13, BH15}.

Bounds on hypercontractive norms of operators are also used to prove expansion of small sets in
graphs.
Indeed, if $f$ is the indicator function of set $S$ of measure $\delta$ in a graph with adjacency
matrix $A$, then we have that for any $p \leq q$,
\[
\Phi(S) 
~=~ 
1 - \frac{\ip{f, A f}}{\norm{2}{f}^2} 
~\geq~
1 - \frac{\norm{q^*}{f} \cdot \norm{q}{Af}}{\delta}
~\geq~
1 - \norm{p}{q}{A} \cdot \delta^{1/p - 1/q} \mper
\]
It was proved by Barak \etal \cite{BBHKSZ12} that the above connection to small-set expansion can in
fact be made two-sided for a special case of the \onorm{2}{q}.
They proved by that  to resolve the promise version of the small-set expansion (SSE) problem, it
suffices to distinguish the cases $\norm{2}{q}{A} \leq c \cdot \sigma_{\min}$ and $\norm{2}{q}{A}
\geq C\cdot \sigma_{\min}$, where $\sigma_{\min}$ is the least non-zero singular value of $A$ and
$C > c > 1$ are appropriately chosen constants based on the parameters of the SSE problem. 
Thus, the approximability of \onorm{2}{q} is closely related to the small-set expansion problem. In
particular,  proving the NP-hardness of approximating the \onorm{2}{q} is (necessarily) an
intermediate goal  towards proving the Small-Set Expansion Hypothesis of Raghavendra and Steurer
\cite{RS10}. 

However,  relatively few results algorithmic and hardness results are known for approximating
hypercontractive norms.  
A result by Steinberg's \cite{Steinberg05} gives an upper bound of $O(\max\inbraces{m,n}^{25/128})$
on the approximation factor, for all $p,q$. 
For the case of \onorm{2}{q} (for any $q > 2$), Barak \etal \cite{BBHKSZ12} give an approximation
algorithm for the promise version of the problem described above, running  
in time $\exp\inparen{\tilde{O}(n^{2/q})}$. 
They also provide an additive approximation algorithm
for the \onorm{2}{4} (where the error depends on \onorm{2}{2} and \onorm{2}{\infty} of
$A$), which was extended to the \onorm{2}{q} by Harrow and Montanaro \cite{HM13}. 
Barak \etal also prove NP-hardness of approximating
$\norm{2}{4}{A}$ within a factor of $1 + \tilde{O}(1/ n^{o(1)})$,  and hardness of approximating
better than $\exp{O((\log n)^{1/2-\eps})}$ in polynomial time, assuming the Exponential Time
Hypothesis (ETH). 
This reduction was also used by Harrow, Natarajan and Wu \cite{HNW16} to prove that $\tilde{O}(\log n)$
levels of the Sum-of-Squares SDP hierarchy cannot approximate $\norm{2}{4}{A}$ within any constant
factor.

It is natural to ask if the bottleneck in proving (constant factor) hardness of approximation for
\onorm{2}{q} arises from the fact from the nature of the domain (the $\ell_2$ ball) or from
hypercontractive nature of the objective.  As discussed in \cref{sec:overview}, \emph{all}
hypercontractive norms present a barrier for gadget reductions, since if a ``true'' solution $x$ is
meant to encode the assignment to a (say) label cover problem with consistency checked via local
gadgets, then (for $q > p$), a ``cheating solution'' may make the value of $\norm{q}{Ax}$ very
large by using a sparse $x$ which does not carry any meaningful information about the underlying
label cover problem.

We show that (somewhat surprisingly, at least for the authors)  it is indeed possible to
overcome the barrier for gadget reductions for hypercontractive norms, for any $2 < p < q$ (and by
duality, for any $p < q< 2$). 
This gives the first NP-hardness result for hypercontractive norms (under randomized reductions). 
Assuming ETH, this also rules out a constant factor approximation algorithm that runs in $2^{n^{\delta}}$
for some $\delta := \delta(p, q)$. 
%
%
%
\begin{theorem}
For any $p, q$ such that $1 < p \leq q < 2$ or $2 < p \leq q < \infty$ and a constant $c > 1$,  it
is NP-hard to approximate \onorm{p}{q} within a factor of $c$. The reduction runs in time $n^{B_p
  \cdot q}$ for $2 < p < q$, where $B_p = \poly(1 / (1 - \gamma_{p^*}))$. 
\label{thm:main_hype_vanilla}
\end{theorem}
We show that the above hardness can be strengthened to any constant factor via a simple tensoring
argument. 
In fact, this also shows that it is hard to approximate $\norm{p}{q}{A}$ within almost polynomial factors
unless NP is in randomized quasi-polynomial time. This is the content of the following theorem.
\begin{theorem}
For any $p, q$ such that $1 < p \leq q < 2$ or $2 < p \leq q < \infty$ and $\epsilon > 0$,  there
is no polynomial time algorithm that approximates the \onorm{p}{q} of an $n \times n$ matrix within
a factor $2^{\log^{1 - \epsilon} n}$ unless $\NP \subseteq \BPTIME{2^{(\log n)^{O(1)}}}$. When $q$
is an even integer, the same inapproximability result holds unless $\NP \subseteq \DTIME{2^{(\log
    n)^{O(1)}}}$
\label{thm:main_hype}
\end{theorem}

We also note that the operator $A$ arising in our reduction in  \cref{thm:main_hype_vanilla}
 satisfies $\sigma_{\min}(A) \approx 1$ (and is in fact a product of a carefully chosen 
projection and a scaled random Gaussian matrix).
For such an $A$, we prove the hardness of distinguishing $\norm{p}{q}{A} \leq c$ and
$\norm{p}{q}{A} \geq C$, for constants $C > c > 1$. 
For the corresponding problem in
the case of \onorm{2}{q}, Barak \etal \cite{BBHKSZ12} gave a subexponential algorithm running in
time $\exp\inparen{O(n^{2/q})}$ (which works for every $C > c > 1$). 
On the other hand, since the running time of our reduction is $n^{O(q)}$, we get that 
assuming ETH, we show that no algorithm can distinguish the above cases for \onorm{p}{q} in
time $\exp\inparen{n^{o(1/q)}}$, for any $p \leq q$ when $2 \notin [p,q]$.

While the above results give some possible reductions for working with hypercontractive norms, it
remains an interesting problem to understand the role of the domain as a barrier to proving hardness
results for the \onorm{2}{q} problems. In fact, no hardness results are available even for the more
general problem of polynomial optimization over the $\ell_2$ ball.
We view the above theorem as providing some evidence that while hypercontractive norms have been
studied as a single class so far, the case when $2 \in [p,q]$ may be qualitatively different (with respect to
techniques) from the case when $2 \notin [p,q]$. 
This is indeed known to be true in the \emph{non-hypercontractive case} with $p \geq q$. In fact,
our results are obtained via new hardness results for the case $p \geq q$, as described below.

\paragraph{The non-hypercontractive case.}
Several results are known in the case when $p \geq q$, and we summarize known results for matrix
norms  in \cref{fig:normbounds}, for the both the hypercontractive and non-hypercontractive cases. 
While the case of $p=q=2$ corresponds to the spectral norm,
the problem is also easy when $q = \infty$ (or equivalently $p = 1$) since this corresponds to
selecting the row of $A$ with the maximum $\ell_{p^*}$ norm. Note that in general,
\cref{fig:normbounds} is symmetric about the principal diagonal. Also note that if $\norm{p}{q}{A}$
is a hypercontractive norm ($p < q$) then so is the equivalent $\norm{q^*}{p^*}{A^T}$ (the
hypercontractive and non-hypercontractive case are separated by the non-principal diagonal).
%

%
\begin{figure}
\begin{center}
\includegraphics[width=3.6in]{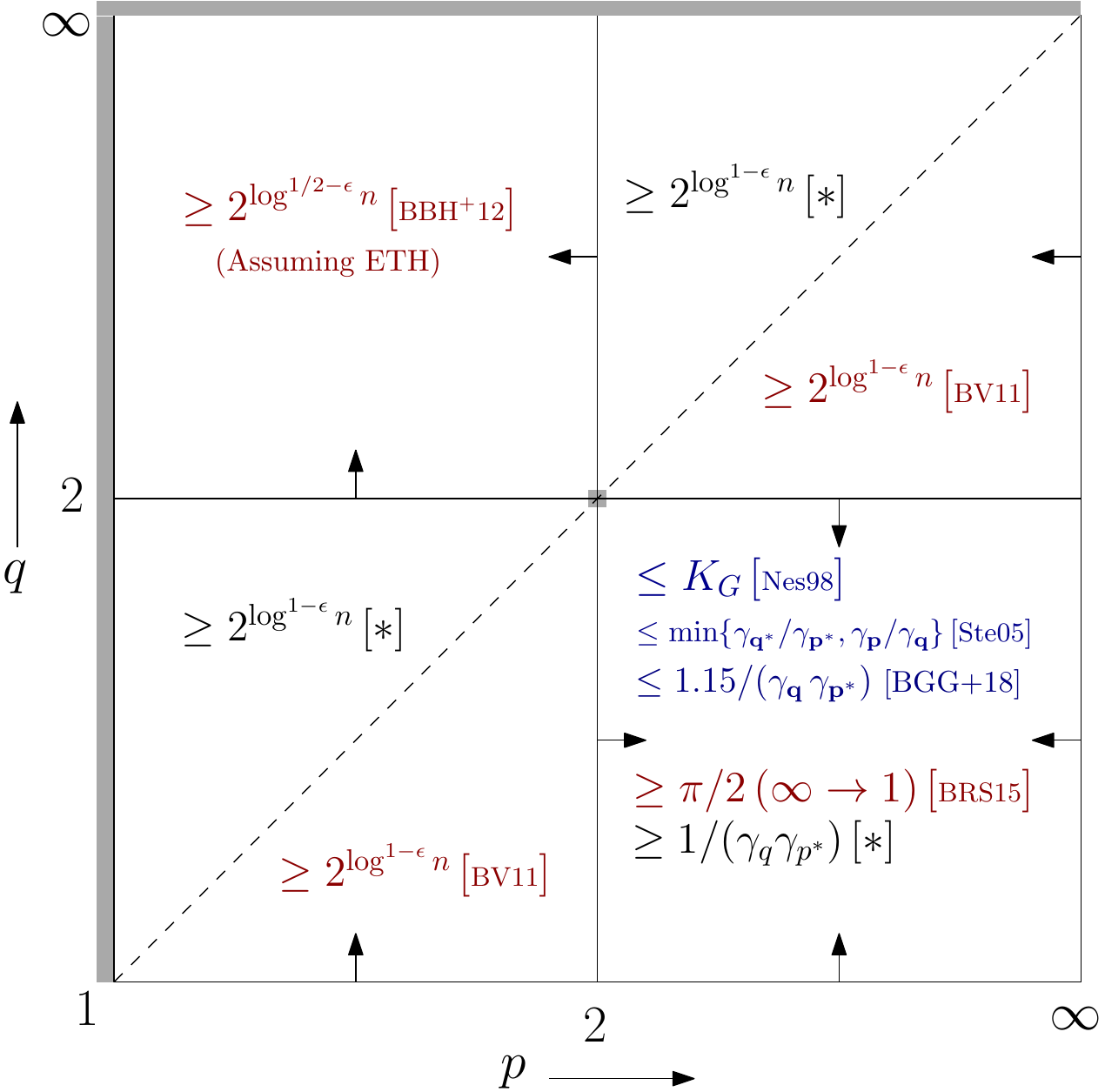}
\end{center}
\vspace{-15 pt}
\caption{Upper and lower bounds for approximating $\|A\|_{p \rightarrow q}$. Arrows
  indicate the region to which a boundary belongs and thicker shaded regions represent exact
  algorithms. Our results are indicated by [$*$]. We omit UGC-based hardness results in the figure.}
\label{fig:normbounds}
\end{figure}

As is apparent from the figure, the problem of approximating $\norm{p}{q}{A}$ for $p \geq q$ 
admits good  approximations when $2 \in [q,p]$, and is hard otherwise. 
For the case when $2 \notin [q,p]$,
an upper bound of $O(\max\{m,n\}^{25/128})$ on the approximation
ratio was proved by Steinberg \cite{Steinberg05}.  
Bhaskara and Vijayaraghavan \cite{BV11} showed NP-hardness of approximation within any
constant factor, and hardness of approximation within an $O\inparen{2^{(\log n)^{1-\eps}}}$ factor
for arbitrary $\eps > 0$ assuming $\NP \not\subseteq \DTIME{2^{(\log n)^{O(1)}}}$.

Determining the right
constants in these approximations when $2 \in [q,p]$ has been of considerable interest in the
analysis and optimization community.
For the case of \onorm{\infty}{1}, Grothendieck's theorem \cite{Grothendieck56} shows that the
integrality gap of a semidefinite programming (SDP) relaxation is bounded by a constant, and the
(unknown) optimal value is now called the Grothendieck constant $K_G$. Krivine \cite{Krivine77} proved
an upper bound of $\pi/(2\ln(1+\sqrt{2})) = 1.782 \ldots$ on $K_G$, and it was later shown by
Braverman \etal that $K_G$ is strictly smaller than this bound. The best known lower bound on $K_G$
is about $1.676$, due to (an unpublished manuscript of) Reeds \cite{Reeds91} (see also \cite{KO09}
for a proof).

An upper bound of $K_G$ on the approximation factor also
follows from the work of Nesterov \cite{Nesterov98} for any $p \geq 2 \geq q$. A later work of
Steinberg \cite{Steinberg05} also gave an upper bound of 
$\min\inbraces{\gamma_p/\gamma_q, \gamma_{q^*}/\gamma_{p^*}}$, where $\gamma_p$ denotes
$p^{th}$  norm of a standard normal random variable (\ie the $p$-th root of the $p$-th Gaussian moment). 
Note that Steinberg's bound is less than 
$K_G$ for some values of $(p,q)$, in particular for all values of the form $(2,q)$ with $q \leq 2$
(and equivalently $(p,2)$ for $p \geq 2$), where it equals $1/\gamma_q$ (and $1/\gamma_{p^*}$ for
$(p,2)$).

On the hardness side, \Briet, Regev and Saket \cite{BRS15} showed NP-hardness of $\pi/2$ for the
\onorm{\infty}{1},  strengthening a hardness result of Khot and Naor based on the Unique Games
Conjecture (UGC) \cite{KN08} (for a special case of the Grothendieck problem when the matrix $A$ is
positive semidefinite). Assuming UGC, a hardness result matching Reeds' lower bound was proved by Khot and
\Odonnell \cite{KO09}, and hardness of approximating within $K_G$ was proved by Raghavendra and
Steurer \cite{RS09}. 

For a related problem known as the $L_p$-Grothendieck problem, where the goal is to maximize
$\ip{x,Ax}$ for $\norm{p}{x} \leq 1$, results by  Steinberg \cite{Steinberg05} and 
Kindler, Schechtman and Naor \cite{KNS10} give an upper bound of $\gamma_p^2$, and a matching lower
bound was proved assuming UGC by \cite{KNS10}, which was strengthened to NP-hardness by Guruswami
\etal \cite{GRSW16}. 
However, note that this problem is quadratic and not necessarily bilinear, and is in
general much harder than the Grothendieck problems considered here. In particular, the case of 
$p = \infty$ only admits an $\Theta(\log n)$ approximation instead of $K_G$ for the bilinear version
\cite{AMMN06, ABHKS05}.

We extend the hardness results of \cite{BRS15} for the $\infty \to 1$
and $2 \to 1$ norms of a matrix to any $p \geq 2 \geq q$. 
The hardness factors obtained match the performance of known algorithms (due to Steinberg
\cite{Steinberg05}) for the cases of $2 \to q$ and $p \to 2$.
\begin{theorem}
    For any $p, q$ such that $\infty \geq p \geq 2 \geq q \geq 1$ and $\epsilon > 0$, 
    it is NP-hard to approximate the \onorm{p}{q} within a factor $1 / (\gamma_{p^*} \gamma_q) - \epsilon$. 
   \label{thm:main_nhc}
\end{theorem}

In subsequent work \cite{BGGLT18b} motivated by the hardness results herein, we also give an 
improved approximation for \onorm{p}{q} when $2 \in [q,p]$ (inspired by the above hardness result) 
which achieves an approximation factor of $C_0\cdot (1 / (\gamma_{p^*} \gamma_q))$, where 
$C_0 \approx 1/(\ln(1+\sqrt{2}))$ is a constant comparable to that arising in Krivine's upper bound on 
the Grothendieck constant \cite{Krivine77}. 

Both \cref{thm:main_hype_vanilla} and \cref{thm:main_nhc} are consequences of a more technical theorem,
which proves hardness of approximating $\norm{2}{r}{A}$ for $r < 2$ (and hence $\norm{r^*}{2}{A}$
for $r^* > 2$) while providing additional structure in the matrix $A$ produced by the
reduction. This is proved in \cref{sec:nhc}.
We also show our methods can be used to provide a simple proof (albeit via randomized reductions)
of the $2^{\Omega((\log n)^{1-\eps})}$ hardness 
for the non-hypercontractive case when $2 \notin [q,p]$, which was
proved by \cite{BV11}. This is presented in \cref{sec:reverse}.

\subsection{Proof Overview}
\label{sec:overview}
\paragraph{The hardness of proving hardness for hypercontractive norms.}
Reductions for  various geometric problems use a ``smooth'' version of the Label Cover problem,
composed with long-code functions for the labels of the variables. In various reductions, including
the ones by Guruswami \etal \cite{GRSW16} and \Briet \etal \cite{BRS15} (which we closely follow) the
solution vector $x$ to the geometric problem consists of the Fourier coefficients of the various
long-code functions, with a ``block'' $x_v$ for each vertex of the label-cover instance. 
The relevant geometric operation (transformation by the matrix $A$ in our case)
consists of projecting to a space which enforces the consistency constraints derived from the
label-cover problem, on the Fourier coefficients of the encodings.

However, this strategy presents with two problems when designing reductions for hypercontractive
norms. Firstly, while projections maintain the $\ell_2$ norm of encodings corresponding to
consistent labelings and reduce that of inconsistent ones, their behaviour is harder to analyze for
$\ell_p$ norms for $p \neq 2$. Secondly, the \emph{global} objective of maximizing $\norm{q}{Ax}$ is
required to enforce different behavior within the blocks $x_v$, than in the full vector $x$. The
block vectors $x_v$ in the solution corresponding to a satisfying assignment of label cover are
intended to be highly sparse, since they correspond to ``dictator functions'' which have only one
non-zero Fourier coefficient. This can be enforced in a test using the fact that for a vector 
$x_v \in \R^t$, $\norm{q}{x_v}$ is a convex function of $\norm{p}{x_v}$ when $p \leq q$, and is
maximized for vectors with all the mass concentrated in a single coordinate.
However, a global objective function which tries to maximize $\sum_v \norm{x_v}_q^q$, 
also achieves a high value from global vectors $x$ which concentrate all the mass on coordinates
corresponding to  few vertices of the label cover instance, and do not carry any meaningful
information about assignments to the underlying label cover problem.
 
Since we can only check for a global objective which is the $\ell_q$ norm of some vector
involving coordinates from blocks across the entire instance, it is not clear how to enforce local
Fourier concentration (dictator functions for individual long codes) and global well-distribution
(meaningful information regarding assignments of most vertices) using the same objective function. 
While the projector $A$ also enforces a linear relation between the block vectors $x_u$ and $x_v$
for all edges $(u,v)$ in the label cover instance, using this to ensure well-distribution across blocks
 seems to require a very high density of constraints in the label cover instance, and no hardness
 results are available in this regime.
\paragraph{Our reduction.}
We show that when $2 \notin [p,q]$, it is possible to bypass the above issues using hardness of
$\norm{2}{r}{A}$ as an intermediate (for $r < 2$). Note that since $\norm{r}{z}$ is a \emph{concave}
function of $\norm{2}{z}$ in this case, the test favors vectors in which the mass is
well-distributed and thus solves the second issue. For this, we use local tests based on
the Berry-\Esseen
theorem (as in \cite{GRSW16} and \cite{BRS15}). Also, since the starting point now is the $\ell_2$
norm, the effect of projections is easier to analyze. This reduction is discussed in \cref{sec:nhc}.

By duality, we can interpret the above as a
hardness result for $\norm{p}{2}{A}$ when $p > 2$ (using $r = p^*$). We then convert this to a
hardness result for \onorm{p}{q} in the hypercontractive case by composing $A$ with an ``approximate
isometry'' $B$ from $\ell_2 \to \ell_q$ (\ie $\forall y~ \norm{q}{By} \approx \norm{2}{y}$) since we
can replace $\norm{2}{Ax}$ with $\norm{q}{BAx}$.
Milman's version of the Dvoretzky theorem \cite{Vershynin17} implies random operators to a
sufficiently high dimensional ($n^{O(q)}$) space satisfy this property, which then yields constant
factor hardness results for the \onorm{p}{q}. A similar application of Dvoretzky's theorem also
appears in an independent work of Krishnan \etal \cite{KMW18} on sketching matrix norms.

We also show that the hardness for hypercontractive norms can be
amplified via tensoring. This was known previously for the \onorm{2}{4} using an argument based
on parallel repetition for QMA \cite{HM13}, and for the case of $p = q$ \cite{BV11}. We give a
simple argument based on convexity, which proves this for all $p \leq q$, but appears to
have gone unnoticed previously. The amplification is then used to prove hardness of approximation
within almost polynomial factors.

\paragraph{Non-hypercontractive norms.}
We also use the hardness of $\norm{2}{r}{A}$ to obtain hardness for the non-hypercontractive
case of $\norm{p}{q}{A}$ with $q < 2 < p$, by using an operator that ``factorizes'' through 
$\ell_2$.
In particular, we obtain hardness results for $\norm{p}{2}{A}$ and $\norm{2}{q}{A}$ (of factors
$1/\gamma_{p^*}$and $1/\gamma_q$ respectively) using the reduction in \cref{sec:nhc}. We
then combine these hardness results using additional properties of the operator $A$ obtained in the
reduction, to obtain a hardness of factor $(1/\gamma_{p^*}) \cdot (1/\gamma_{q})$ for the
\onorm{p}{q} for $p > 2 > q$. The composition, as well as the hardness
results for hypercontractive norms, are presented in \cref{sec:hypercontractive}.

We also obtain a simple proof of the $2^{\Omega((\log n)^{1-\eps})}$ hardness for the
non-hypercontractive case when $2 \notin [q,p]$ (already proved by Bhaskara and Vijayaraghavan
\cite{BV11}) via an approximate isometry argument as used in the hypercontractive case. In the
hypercontractive case, we started from a constant factor hardness of the \onorm{p}{2} and the same
factor for \onorm{p}{q} using the fact that for a random Gaussian matrix $B$ of appropriate
dimensions, we have $\norm{q}{Bx} \approx \norm{2}{x}$ for all $x$. We then amplify the hardness via
tensoring. 
In the non-hypercontractive case, we start with a hardness for \onorm{p}{p} (obtained via the above
isometry), which we \emph{first} amplify via tensoring. We then apply another approximate isometry
result due to Schechtman \cite{Schechtman87}, which gives a samplable distribution $\calD$ over
random matrices $B$ such that with high probability over $B$, we have $\norm{q}{Bx} \approx
\norm{p}{x}$ for all $x$.

\bigskip

We thus view the above results as showing that combined with a basic hardness for \onorm{p}{2}, the
basic ideas of duality, tensoring, and embedding (which builds on powerful results from functional
analysis) can be combined in powerful ways to prove strong results in both the hypercontractive and
non-hypercontractive regimes. 



\section{Preliminaries and Notation}\label{sec:prelims}
\ifnotconf
\subsection{Matrix Norms}
\fi
For a vector $x\in\R^n$, throughout this paper we will use $x(i)$ to denote its $i$-th coordinate. 
For $p \in [1, \infty)$, we define $\cnorm{p}{\cdot}$ to denote the counting $p$-norm and $\enorm{p}{\cdot}$
to denote the expectation $p$-norm; \ie for a vector $x\in\R^n$, 
\[
    \cnorm{p}{x} := \inparen{\sum_{i\in [n]} |x(i)|^{p}}^{1/p}
 \quad \mbox{ and } \quad 
    \enorm{p}{x} := \Ex{i\sim [n]}{|x(i)|^p}^{1/p} =\inparen{\frac1n \cdot \sum_{i\in [n]} |x(i)|^{p}}^{1/p}.
\]
Clearly $\cnorm{p}{x} = \enorm{p}{x}\cdot n^{1/p}$. 
For $p = \infty$, we define $\cnorm{\infty}{x} = \enorm{\infty}{x} := \max_{i \in [n]} |x(i)|$. 
We will use $p^*$ to 
denote the `dual' of $p$, i.e. $p^* = p/(p-1)$. 
Unless stated otherwise, we usually work with $\cnorm{p}{\cdot}$. 
We also define inner product $\cip{x,y}$ to denote the inner product under the counting measure unless stated
otherwise; \ie for two vectors $x, y \in \R^n$,
$\cip{x, y} := \sum_{i\in [n]} x(i)y(i)$.

We next record a well-known fact about $p$-norms that is used 
in establishing many duality statements.  

\begin{observation}
\label{p-norm:dual}
    For any $p\in [1,\infty]$, 
$\cnorm{p}{x} = \sup_{\cnorm{p^*}{y}=1} \,\mysmalldot{y}{x}$.
\end{observation}

We next define the primary problems of interest in this paper. 
\begin{definition}
    For $p,q\in [1,\infty]$, the \onorm{p}{q} problem is to maximize 
    \[
        \frac{\cnorm{q}{Ax}}{\cnorm{p}{x}}
    \]
    given an $m\times n$ matrix $A$. 
\end{definition}

\begin{definition}
    For $p,q\in [1,\infty]$, we define a  generalization of the Grothendieck problem, namely \groth{p}{q}, as the 
    problem of computing 
    \[
        \sup_{\cnorm{p}{y}=1} \,\sup_{\cnorm{q}{x}=1} \mysmalldot{y}{Ax}
    \]
    given an $m\times n$ matrix $A$. 
\end{definition}

The original Grothendieck problem is precisely \groth{\infty}{\infty}. 
We next state the well known equivalence of \onorm{p}{q}, \groth{q\textsuperscript{*}}{p}, and 
\onorm{q\textsuperscript{*}}{p\textsuperscript{*}}.
\begin{observation}
\label{onorm:groth}
    For any $p,q\in [1,\infty]$ and any matrix $A$,
\[
\cnorm{p}{q}{A} =\sup_{\cnorm{q^*}{y}=1}\,\sup_{\cnorm{p}{x}=1}\mysmalldot{y}{Ax} = \cnorm{q^*}{p^*}{A^T}.
\]
\end{observation}
\begin{proof}
Using $\mysmalldot{y}{Ax} = \mysmalldot{x}{A^Ty}$, 
\begin{align*}
        \cnorm{p}{q}{A}
&=        \sup_{\cnorm{p}{x}=1}\cnorm{q}{Ax}
=         \sup_{\cnorm{p}{x}=1}\,\sup_{\cnorm{q^*}{y}=1}\mysmalldot{y}{Ax} 
=        \sup_{\cnorm{q^*}{y}=1}\,\sup_{\cnorm{p}{x}=1}\mysmalldot{y}{Ax} \\
&=    
        \sup_{\cnorm{p}{x}=1}\,\sup_{\cnorm{q^*}{y}=1}\mysmalldot{x}{A^Ty} 
        =
        \sup_{\cnorm{q^*}{y}=1}\cnorm{p^*}{A^Ty} 
	=
        \cnorm{q^*}{p^*}{A^T}\mper 
\tag*{\qedhere}
    \end{align*}
\end{proof}

The following observation will be useful for composing hardness maps for \onorm{p}{2} and \onorm{2}{q} 
to get \onorm{p}{q} hardness for when $p>q$ and $p\geq 2 \geq q$. 
\begin{observation}
\label{composition:soundness}
    For any $p,q,r\in [1,\infty]$ and any matrices $B, C$, 
\[
        \cnorm{p}{q}{BC} 
        =
        \sup_x \frac{\cnorm{q}{BCx}}{\cnorm{p}{x}}
        \leq 
        \sup_x \frac{\cnorm{r}{q}{B} \cnorm{r}{Cx}}{\cnorm{p}{x}} 
\leq 
        \cnorm{r}{q}{B} \cnorm{p}{r}{C}.
\]
\end{observation}


\ifnotconf
\subsection{Fourier Analysis}
\label{sec:fourier}
We introduce some basic facts about Fourier analysis of Boolean functions.
Let $R \in \N$ be a positive integer, and consider a function $f : \{ \pm 1 \}^R \to \R$. 
For any subset $S \subseteq [R]$ let $\chi_S := \prod_{i \in S} x_i$. 
Then we can represent $f$ as 
\begin{equation}
\label{eq:inverse_fourier}
f(x_1, \dots, x_R) = \sum_{S \subseteq [R]} \hatf(S) \cdot \chi_S(x_1, \dots x_R),
\end{equation}
where 
\begin{equation}
\label{eq:fourier}
\hatf(S) = \E_{x \in \{ \pm 1 \}^R} [f(x) \cdot \chi_S(x)] \mbox{ for all } S \subseteq [R].
\end{equation}
The {\em Fourier transform} refers to a linear operator $F$ that maps $f$ to $\hatf$ as defined as~\eqref{eq:fourier}. 
We interpret $\hatf$ as a $2^R$-dimensional vector whose coordinates are indexed by $S \subseteq [R]$. 
Endow the expectation norm and the expectation norm to $f$ and $\hatf$ respectively; i.e., 
\[
\enorm{p}{f} := \left( \Ex{x \in \{ \pm 1 \}^R}{|f(x)|^p} \right)^{1/p}
\quad \mbox{ and } \quad 
\cnorm{p}{\hatf} := \left( \sum_{S \subseteq [R]} | \hatf(S) |^p \right)^{1/p}.
\]
as well as the corresponding inner products $\mysmalldot{f}{g}$ and $\mysmalldot{\hatf}{\hatg}$ consistent with their $2$-norms.
We also define the {\em inverse Fourier transform} $F^T$ to be a linear operator 
that maps a given $\hatf : 2^R \to \R$ to $f : \{ \pm 1 \}^R \to \R$ defined as in~\eqref{eq:inverse_fourier}. 
We state the following well-known facts from Fourier analysis. 
\begin{observation} [Parseval's Theorem]
For any $f : \{ \pm 1 \}^R \to \R$, $\enorm{2}{f} = \cnorm{2}{F f}$.
\end{observation}
\begin{observation} $F$ and $F^T$ form an adjoint pair; i.e., for any $f : \{ \pm 1 \}^R \to \R$
and $\hatg : 2^R \to \R$, 
\[ 
\mysmalldot{\hatg}{Ff} = 
\mysmalldot{F^T \hatg}{f}.
 \]
\end{observation}
\begin{observation}
$F^T F$ is the identity operator.
\end{observation}

In~\cref{sec:nhc}, we also consider a {\em partial} Fourier transform $F_P$
that maps a given function $f : \{ \pm 1 \}^R \to \R$ to a vector $\hatf : [R] \to \R$
defined as $\hatf(i) = \E_{x \in \{ \pm 1 \}^R} [f(x) \cdot x_i]$ for all $i \in [R]$.
It is the original Fourier transform where $\hatf$ is further projected to $R$ coordinates corresponding to linear coefficients. 
The partial inverse Fourier transform $F_P^T$ is a transformation that maps 
a vector $\hatf : [R] \to \R$ to a function $f : \{ \pm 1 \}^R \to \R$ as in~\eqref{eq:inverse_fourier} restricted to $S = \{ i \}$ for some $i \in [R]$. 
These partial transforms satisfy similar observations as above: 
(1) $\enorm{2}{f} \geq \cnorm{2}{F_P f}$, 
(2) $\enorm{2}{F_P^T \hatf} = \cnorm{2}{\hatf}$,
(3) $F_P$ and $F_P^T$ form an adjoint pair, 
and (4) $(F_P^T F_P) f = f$ if and only if $f$ is a linear function. 

\subsection{Smooth Label Cover}
An instance of \labelcover is given by a quadruple $\calL = (G, [R], [L], \Sigma)$ that consists of a regular connected graph $G = (V, E)$, a label set $[R]$ for some positive integer $n$, and a collection $\Sigma = ((\pi_{e, v}, \pi_{e, w}) : e = (v, w) \in E)$ of pairs of maps both from $[R]$ to $[L]$ associated with the endpoints of the edges in $E$. Given a {\em labeling} $\ell : V \to [R]$, we say that an edge $e = (v, w) \in E$ is {\em satisfied } if $\pi_{e, v}(\ell(v)) = \pi_{e, w}(\ell(w))$. Let $\opt(\calL)$ be the maximum fraction of satisfied edges by any labeling. 

The following hardness result for \labelcover, given in~\cite{GRSW16}, is a slight variant of the original construction due to~\cite{Khot02}. The theorem also describes the various structural properties, including smoothness, that are identified by the hard instances.
\begin{theorem}
\label{thm:smooth_label_cover}
For any $\xi >0$ and $J \in \N$, there exist positive integers $R = R(\xi, J), L = L(\xi, J)$ and $D = D(\xi)$, and a \labelcover instance $(G, [R], [L], \Sigma)$ as above such that
\begin{itemize}
\item (Hardness): It is NP-hard to distinguish between the following two cases:
\begin{itemize}
\item (Completeness): $\opt(\calL) = 1$. 
\item (Soundness): $\opt(\calL) \leq \xi$.
\end{itemize}

\item (Structural Properties): 
\begin{itemize}
\item ($J$-Smoothness): For every vertex $v \in V$ and distinct $i, j \in [R]$, we have 
\[
\Pr{e : v \in e}{\pi_{e,v}(i)=\pi_{e,v}(j)} \leq 1 / J.
\]
\item ($D$-to-$1$): For every vertex $v \in V$, edge $e \in E$ incident on $v$, and $i \in [L]$, we have $|\pi^{-1}_{e, v}(i)| \leq D$; that is at most $D$ elements in $[R]$ are mapped to the same element in $[L]$. 

\item (Weak Expansion): For any $\delta > 0$ and vertex set $V' \subseteq V$ such that $|V'| = \delta \cdot |V|$, the number of edges among the vertices in $|V'|$ is at least $(\delta^2 / 2) |E|$. 
\end{itemize}

\end{itemize} 
\end{theorem}
\fi

\section{Hardness of \onorm{2}{r} with $r < 2$}\label{sec:nhc}
%
This section proves the following theorem that serves as a starting point of our hardness results. 
The theorem is stated for the expectation norm for consistency with the current literature, but the same statement holds for the counting norm, since if 
$A$ is an $n \times n$ matrix,  $\cnorm{2}{r}{A} = n^{1/r -1/2} \cdot \enorm{2}{r}{A}$. 
Note that the matrix $A$ used in the reduction below does not depend on $r$.
\begin{theorem} For any $\varepsilon > 0$, 
    there is a polynomial time reduction that takes a 3-CNF formula $\phi$ and 
    produces a symmetric matrix $A \in \R^{n \times n}$ with $n = |\phi|^{\poly(1 / \eps)}$ such that 
	\begin{itemize}
        \item (Completeness) If $\phi$ is satisfiable, there exists $x \in
            \R^{n}$ with $|x(i)| = 1$ for all $i \in [n]$ and $Ax = x$.  In particular,
            $\enorm{2}{r}{A} \geq 1$ for all $1 \leq r \leq \infty$.
        \item (Soundness) $\enorm{2}{r}{A} \leq \gamma_r + \varepsilon^{2 - r}$ for all $1 \leq r < 2$. 
	\end{itemize}
    \label{thm:brs}
\end{theorem}
We adapt the proof by \Briet, Regev and Saket for the hardness of $2 \to 1$ and $\infty \to 1$
norms to prove the above theorem. A small difference is that, 
unlike their construction which starts with a Fourier encoding of the long-code functions, we start
with an evaluation table (to ensure that the resulting matrices are symmetric). We also analyze
their dictatorship tests for the case of fractional $r$. 
\ifconf
The details can be found in the full version of the paper.
\fi

\ifnotconf
\subsection{Reduction and Completeness}
Let $\calL = (G, [R], [L], \Sigma)$ be an
instance of \labelcover with $G = (V, E)$. 
In the rest of this section, $n = |V|$ and our reduction will construct a self-adjoint linear operator 
$\bfA : \R^N \to \R^N$ with $N = |V| \cdot 2^R$, which yields a symmetric $N \times N$ matrix
representing $\bfA$ in the standard basis. 
This section concerns the following four Hilbert spaces based on the standard Fourier
analysis composed with $\calL$.
\begin{enumerate}
    \item Evaluation space $\R^{2^R}$. Each function in this space is
        denoted by $f : \{ \pm 1 \}^R \to \RRR$.  The inner product is defined as
        $\ip{f, g} := \Ex{x \in \{ \pm 1 \}^R}{f(x) g(x)}$, which
        induces $\norm{2}{f} := \enorm{2}{f}$. We also define $\enorm{p}{f} := \Ex{x}{|f(x)|^p}^{1/p}$ in this space. 
    \item Fourier space $\R^{R}$. Each function in this space is denoted by $\hatf : [R]
        \to \RRR$.  The inner product is defined as $\mysmalldot{\hatf}{\hatg} := \sum_{i
        \in [R]} \hatf(i) \hatg(i)$, which induces $\norm{2}{\hatf} := \cnorm{2}{\hatf}$. 

    \item Combined evaluation space $\R^{V \times 2^R}$. Each function in this space is
        denoted by $\bff : V \times \{ \pm 1 \}^R \to \RRR$.  The inner product is defined as
        $\mysmalldot{\bff}{\bfgg} := \Ex{v \in V}{\Ex{x \in \{ \pm 1 \}^R}{\bff(v, x)
        \bfgg(v, x)}}$, which induces $\enorm{2}{\bff} := \enorm{2}{\bff}$. 
We also define $\norm{p}{\bff} := \Ex{v, x}{|\bff(v, x)|^p}^{1/p}$ in this space.
    \item Combined Fourier space $\R^{V \times R}$. Each function in this space is denoted
        by $\bfhatf : V \times [R] \to \RRR$.  The inner product is defined as
        $\mysmalldot{\bfhatf}{\bfhatg} := \Ex{v \in V}{\sum_{i \in [R]} \bfhatf(v, i)
    \bfhatg(v, i) }$, which induces $\norm{2}{\bfhatf}$, which is neither a counting nor an
    expectation norm.
\end{enumerate}
Note that $\bff \in \R^{V \times 2^R}$ and a vertex
$v \in V$ induces $f_v \in \R^{2^R}$ defined by $f_v(x) := \bff(v, x)$, and similarly
$\bfhatf \in \R^{V \times R}$ and a vertex $v \in V$ induces $\hatf_v \in \R^{R}$
defined by $\hatf_v(x) := \bfhatf(v, x)$.  
As defined in \cref{sec:fourier}, 
we use the standard
following (partial) {\em Fourier transform} $F$ that maps $f \in \R^{2^R}$ to $\hatf \in
\R^{R}$ as follows.  \footnote{We use only {\em linear Fourier coefficients} in this
work. $F$ was defined as $F_P$ in~\cref{sec:fourier}.} \begin{equation} \hatf(i) = (Ff) (i) := \Ex{x \in \{\pm 1 \}^R}{x_i f(x)}.  \end{equation} 

The (partial) {\em inverse Fourier
transform} $F^T$ that maps $\hatf \in \R^{R}$ to $f \in \R^{2^R}$ is defined by
\begin{equation} f(x) = (F^T \hatf)(x) := \sum_{i \in [R]} x_i \hatf(i). \end{equation}

This Fourier transform can be naturally extended to combined spaces by defining $\bfF :
\bff \mapsto \bfhatf$ as $f_v \mapsto \hatf_v$ for all $v \in V$.  Then $\bfF^T$ maps
$\bfhatf$ to $\bff$ as $\hatf_v \mapsto f_v$ for all $v \in V$.

Finally, let $\bfhatP : \R^{V \times R} \to \R^{V
\times R}$ be the orthogonal projector to the following subspace of the combined
Fourier space: 
\begin{equation} 
    \bfhatL := \left\{ \bfhatf \in \R^{V \times R} :
    \sum_{j \in \pi^{-1}_{e, u}(i)} \hatf_u(i) = \sum_{j \in \pi^{-1}_{e, v}(i)}
    \hatf_v(j) \mbox{ for all } (u, v) \in E \mbox{ and } i \in [L] \right\}\mper
\end{equation}
Our transformation $\bfA : \R^{V \times 2^R} \to \R^{V \times 2^R}$ is
defined by \begin{equation} \bfA := (\bfF^T) \bfhatP \bfF.  \end{equation} In other
words, given $\bff$, we apply the Fourier transform for each $v \in V$, project the
combined Fourier coefficients to $\bfhatL$ that checks the \labelcover consistency, and
apply the inverse Fourier transform. Since $\bfhatP$ is a projector, $\bfA$ is
self-adjoint by design. 

We also note that a similar reduction that produces $(\bfF^T) \bfhatP$ 
was used in Guruswami \etal~\cite{GRSW16} and \Briet \etal~\cite{BRS15} 
for subspace approximation and Grothendieck-type problems, 
and indeed this reduction suffices for~\cref{thm:brs}
except the self-adjointness and additional properties in the completeness case. 

\paragraph{Completeness.}
We prove the following lemma for the completeness case. 
A simple intuition is that if $\calL$ admits a good labeling, 
we can construct a $\bff$ such that each $f_v$ is a linear function and 
$\bfhatf$ is already in the subspace $\bfhatL$. 
Therefore, each of Fourier transform, projection to $\bfhatL$, and inverse Fourier transform 
does not really change $\bff$.

\begin{lemma} [Completeness] Let $\ell : V \to [R]$ be a labeling that satisfies every
    edge of $\calL$. There exists a function $\bff \in \R^{V \times 2^R}$ such that
    $\bff(v, x)$ is either $+1$ or $-1$ for all $v \in V, x \in \{ \pm 1 \}^R$ and
    $\bfA \bff = \bff$.
\end{lemma}
\begin{proof} Let $\bff(v, x) := x_{\ell(v)}$ for
    every $v \in V, x \in \{ \pm 1 \}^R$.  Consider $\bfhatf = \bfF \bff$.  For each
    vertex $v \in V$, $\bfhatf(v, i) = \hatf_v(i) = 1$ if $i = \ell(v)$ and $0$
    otherwise.  Since $\ell$ satisfies every edge of $\calL$, $\bfhatf \in \bfhatL$ and
    $\bfhatP \bfhatf = \bfhatf$.  Finally, since each $f_v$ is a linear function, the
    partial inverse Fourier transform $F^T$ satisfies $(F^T) \hatf_v = f_v$, which
    implies that $(\bfF^T) \bfhatf = \bff$. Therefore, $\bfA \bff = (\bfF^T \bfhatP
    \bfF) \bff = \bff$.
\end{proof}

\subsection{Soundness}
We prove the following soundness lemma. 
This finishes the proof of~\cref{thm:brs} since~\cref{thm:smooth_label_cover} 
guarantees NP-hardness of \labelcover for arbitrarily small $\xi > 0$ and arbitrarily large $J \in \N$. 

\begin{lemma}[Soundness] \label{lem:soundness} For every $\varepsilon > 0$, 
	there exist $\xi > 0$ (that determines $D = D(\xi)$ as in~\cref{thm:smooth_label_cover}) 
	and $J \in \N$ such that if $\opt(\calL) \leq \xi$, $\calL$ is $D$-to-$1$, and $\calL$ is $J$-smooth, 
	$\enorm{2}{r}{\bfA} \leq \gamma_{r} + 4\varepsilon^{2-r}$ for every $1 \leq r < 2$. 
\end{lemma}

\begin{proof}
Let $\bff \in \R^{V \times 2^R}$ be an arbitrary vector such that $\enorm{2}{\bff} =
1$.  Let $\bfhatf = \bfF \bff$, $\bfhatg = \bfhatL \bfhatf$, and $\bfgg = \bfF^T
\bfhatg$ so that $\bfgg = (\bfF^T \bfhatL \bfF) \bff = \bfA \bff$.  By Parseval's
theorem, $\cnorm{2}{\hatf_v} \leq \enorm{2}{f_v}$ for all $v \in V$ and
$\norm{2}{\bfhatf} \leq \enorm{2}{\bff} \leq 1$.  Since $\bfhatL$ is an orthogonal
projection, $\norm{2}{\bfhatg} \leq \norm{2}{\bfhatf} \leq 1$.  
Fix $1 \leq r < 2$ and suppose
\begin{equation}
    \enorm{r}{\bfgg}^r = \Ex{v \in V}{\enorm{r}{g_v}^r} \geq \gamma_r^r +
    4 \epsilon^{2 - r}\mper
    \label{eq:soundness}
\end{equation}

Use~\cref{lem:dict} to obtain $\delta = \delta(\epsilon)$ such that $\enorm{p}{g_v}^p >
(\gamma_p^p + \epsilon) \cnorm{2}{\hatg_v}^p$ implies $\cnorm{4}{\hatg} > \delta
\cnorm{2}{\hatg}$ for all $1 \leq p < 2$ (so that $\delta$ does not depend on $r$), 
and consider \begin{equation} V_0 := \{ v \in V : \cnorm{4}{\hatg_v}
    > \delta \epsilon \mbox{ and } \cnorm{2}{\hatg_v} \leq 1/\epsilon \}.
    \label{eq:vzero} \end{equation} We prove the following lemma that lower bounds the
size of $V_0$.
\begin{lemma}
    For $V_0 \subseteq V$ defined as in~\eqref{eq:vzero}, we have $|V_0| \geq \epsilon^2
    |V|$.
\end{lemma}
\begin{proof} The proof closely follows the proof of Lemma 3.4 of~\cite{BRS15}.
    Define the sets
    \begin{align*}
        V_1 &= \{ v \in V : \cnorm{4}{\hatg_v}
        \leq \delta \epsilon \mbox{ and } \cnorm{2}{\hatg_v} < \epsilon \}, \\ V_2 &= \{ v
        \in V : \cnorm{4}{\hatg_v} \leq \delta \epsilon \mbox{ and } \cnorm{2}{\hatg_v}
        \geq \epsilon \}, \\ V_3 &= \{ v \in V : \cnorm{2}{\hatg_v} > 1/\epsilon \}.
    \end{align*}
    From~\eqref{eq:soundness}, we have
    \begin{equation}
        \sum_{v \in V_0} \enorm{r}{g_v}^r + \sum_{v \in V_1} \enorm{r}{g_v}^r + \sum_{v \in V_2}
        \enorm{r}{g_v}^r + \sum_{v \in V_3} \enorm{r}{g_v}^r \geq (\gamma_r^r + 4
        \epsilon^{2 - r}) |V|\mper
        \label{eq:soundness_sum}
    \end{equation}
    We bound the four sums on the left side of~\eqref{eq:soundness_sum} individually.
    Parseval's theorem and the fact that $r < 2$ implies $\enorm{r}{g_v} \leq
    \enorm{2}{g_v}=\cnorm{2}{\hatg_v}$, and since $\cnorm{2}{\hatg_v} \leq 1/\epsilon$
    for every $v \in V_0$, the first sum in~\eqref{eq:soundness_sum} can be bounded by
    \begin{equation}
        \sum_{v \in V_0} \enorm{r}{g_v}^r \leq |V_0| / \epsilon^r.
        \label{eq:soundness_zero}
    \end{equation}
    Similarly, using the definition of $V_1$
    the second sum in~\eqref{eq:soundness_sum} is at most $\epsilon^r |V|$.
    By~\cref{lem:dict}, for each $v \in V_2$, we have $\enorm{r}{g_v}^r \leq (\gamma_r^r +
    \epsilon) \cnorm{2}{\hatg_v}^r$. Therefore, the third sum in~\eqref{eq:soundness_sum}
    is bounded as
    \begin{align}
        \sum_{v \in V_2} \enorm{r}{g_v}^r &\leq (\gamma_r^r +
            \epsilon) \sum_{v \in V_2} \cnorm{2}{\hatg_v}^r    &   \nonumber \\
        &= (\gamma_r^r + \epsilon) |V_2| \E_{v \in V_2} [ \cnorm{2}{\hatg_v}^r ]
            &     \nonumber \\
        &\leq (\gamma_r^r + \epsilon) |V_2| \E_{v \in V_2} [ \cnorm{2}{\hatg_v}^2 ]^{r / 2} &&
            \mbox{(By Jensen using $r < 2$)}    \nonumber \\
        &= (\gamma_r^r + \epsilon) |V_2| \bigg( \frac{\sum_{v \in V_2}
            \cnorm{2}{\hatg_v}^2 }{|V_2|} \bigg) ^{r / 2} &\nonumber \\
        &\leq (\gamma_r^r + \epsilon) |V_2|^{1 - r / 2} |V|^{r / 2} &&
            (\sum_{v \in V_2} \cnorm{2}{\hatg_v}^2 \leq \sum_{v \in V} \cnorm{2}{\hatg_v}^2
             \leq |V|) \nonumber \\
        &\leq (\gamma_r^r + \epsilon) |V|. &
    \end{align}
    Finally, the fourth sum in~\eqref{eq:soundness_sum} is bounded by
    \begin{align}
        \sum_{v \in V_3} \enorm{r}{g_v}^r &\leq \sum_{v \in V_3} \enorm{2}{g_v}^r 
            && \mbox{(Since $r < 2$)}\nonumber \\
        &= \sum_{v \in V_3} \cnorm{2}{\hatg_v}^r    && \mbox{(By Parseval's
           theorem)}   \nonumber \\
        &= \sum_{v \in V_3} \cnorm{2}{\hatg_v}^{r - 2}\cnorm{2}{\hatg_v}^2 \nonumber \\
        &< \sum_{v \in V_3} \epsilon^{2 - r}\cnorm{2}{\hatg_v}^2 
            && (\cnorm{2}{\hatg_v} > 1/\epsilon \mbox{ for } v \in V_3,\mbox{ and } r < 2)
            \nonumber \\
        &= \epsilon^{2 - r} \sum_{v \in V_3}\cnorm{2}{\hatg_v}^2 \leq \epsilon^{2- r}|V|.&&
    \end{align} Combining the above with~\eqref{eq:soundness_sum} yields
    \begin{align}
        |V_0| & \geq \epsilon^{r} \sum_{v\in V_0} \enorm{r}{g_v}^r  \nonumber \\
        & \geq \epsilon^{r} \bigg( (\gamma_r^r + 4\epsilon^{2 - r}) |V| - 
            \epsilon^r |V| - (\gamma_r^r + \epsilon)|V| - \epsilon^{2-r} |V| \bigg)
            \nonumber \\
        & \geq \epsilon^{r} \epsilon^{2 - r} |V|  =   \epsilon^2 |V|, 
    \end{align} where the last inequality uses the fact that
    $\epsilon^{2- r } \geq \epsilon \geq \epsilon^r$.
\end{proof}
Therefore, $|V_0| \geq \epsilon^2 |V|$ and every vertex of $v$ satisfies 
$\cnorm{4}{\hatg_v} > \delta \epsilon$ and $\cnorm{2}{\hatg_v} \leq 1 / \epsilon$.
Using only these two facts
together with $\bfhatg \in \bfhatL$, \Briet \etal~\cite{BRS15} proved that if the
smoothness parameter $J$ is large enough given other parameters, $\calL$
admits a labeling that satisfies a significant fraction of edges.
\begin{lemma} [Lemma
    3.6 of \cite{BRS15}] Let $\beta := \delta^2 \epsilon^3$.  There exists an absolute
    constant $c' > 0$ such that if $\calL$ is $T$-to-$1$ and $T / (c' \epsilon^8
    \beta^4)$-smooth for some $T \in \N$, there is a labeling that satisfies at least $\epsilon^8 \beta^4
    /1024$ fraction of $E$.
\end{lemma}
This finishes the proof of~\cref{lem:soundness} by setting $\xi := \epsilon^8 \beta^4 /1024$
and $J := D(\xi) / (c' \epsilon^8 \beta^4)$ with $D(\xi)$ defined in~\cref{thm:smooth_label_cover}. 
Given a $3$-SAT formula, $\phi$, 
by the standard property of Smooth Label Cover, the size of the reduction is 
$|\phi|^{O(J \log (1 / \xi))} = |\phi|^{\poly(1/\eps)}$. 
\end{proof}
\fi

\section{Hardness of \onorm{p}{q}}\label{sec:hypercontractive}
In this section, we prove our main results. 
We prove \cref{thm:main_nhc} on hardness of approximating \onorm{p}{q} when $p \geq 2 \geq q$, and 
\cref{thm:main_hype} on hardness of approximating \onorm{p}{q} when 
$2 < p < q$. By duality, the same hardness is implied for the case of $p < q < 2$.  

Our result for $p \geq 2 \geq q$ in~\cref{sec:composition} follows 
from \cref{thm:brs} using additional properties in the completeness case. 
For hypercontractive norms, we start by showing 
constant factor hardness via reduction from \onorm{p}{2} (see \cref{isometry}), and then amplify the 
hardness factor by using the fact that all hypercontractive norms productivize under Kronecker product, 
which we prove in \cref{sec:productivize}. 
\subsection{Hardness for $p \geq 2 \geq q$}
\label{sec:composition}

We use~\cref{thm:brs} to prove hardness of \onorm{p}{q} for $p \geq 2 \geq q$, which proves~\cref{thm:main_nhc}.

\begin{proofof}{\cref{thm:main_nhc}}
	Fix $p, q$, and $\delta > 0$ such that $\infty \geq p \geq 2 \geq q$ and $p > q$. 
	Our goal is to prove that \onorm{p}{q} is NP-hard to approximate within a factor $1 / (\gamma_{p^*}\gamma_q + \delta)$. 
	For \onorm{2}{q} for $1 \leq q < 2$, Theorem~\ref{thm:brs} (with $\varepsilon \leftarrow \delta^{1/(2-q)}$) directly proves a hardness
     ratio of $1/(\gamma_q + \varepsilon^{2-q}) = 1/(\gamma_q + \delta)$. 
	By duality, it also gives an $1/(\gamma_{p^*} + \delta)$ hardness for \onorm{p}{2} for $p > 2$. 

    For \onorm{p}{q} for $p > 2 > q$,
	apply~\cref{thm:brs} with $\varepsilon = (\delta / 3)^{\max(1/(2-p^*), 1/(2-q))}$. 
	It gives a polynomial time reduction that produces a symmetric matrix $A \in \R^{n \times n}$ given a \threesat formula $\phi$. 
Our instance for \onorm{p}{q} is  $AA^T = A^2$.
    \begin{itemize}
        \item (Completeness) If $\phi$ is satisfiable, there exists $x \in \R^{n}$ such that $|x(i)| = 1$ for all $i \in [N]$ and $Ax = x$. Therefore, $A^2
            x = x$ and $\enorm{p}{q}{A^2} \geq 1$. 

        \item (Soundness) If $\phi$ is not satisfiable, 
\begin{align*}
\enorm{p}{2}{A} &~=~ \enorm{2}{p^*}{A} \leq  \gamma_{p^*} + \varepsilon^{2-p^*} \leq  \gamma_{p^*} +
                  \delta / 3, 
\mbox{ and } \\
\enorm{2}{q}{A} &~\leq~ \gamma_{q} + \varepsilon^{2-q} \leq \gamma_{q} + \delta / 3. 
\end{align*}
This implies that 
            \[ \enorm{p}{q}{A^2} \leq \enorm{p}{2}{A}  \enorm{2}{q}{A} \leq
                (\gamma_{p^*} + \delta /3)(\gamma_{q} + \delta / 3) \leq
                \gamma_{p^*}\gamma_q + \delta \mper \]
    \end{itemize} 
This creates a gap of $1 / (\gamma_{p^*} \gamma_q + \delta)$ between the completeness and the soundness case. The same gap holds for the counting norm since $\cnorm{p}{q}{A^2} = n^{1/q - 1/p} \cdot \enorm{p}{q}{A^2}$. 
\end{proofof}

\subsection{Reduction from \onorm{p}{2} via Approximate Isometries}
\label{isometry}
Let $A \in \R^{n \times n}$ be a hard instance of \onorm{p}{2}. For any $q \geq 1$, if a matrix 
$B \in \R^{m \times n}$ satisfies $\cnorm{q}{Bx} = (1 \pm o(1)) \cnorm{2}{x}$ for all $x \in \R^n$,  
then $\norm{p}{q}{BA} = (1 \pm o(1)) \norm{p}{2}{A}$. Thus $BA$ will serve as a hard instance for 
\onorm{p}{q} if one can compute such a matrix $B$ efficiently. In fact, a consequence of the 
Dvoretzky-Milman theorem is that a sufficiently tall random matrix $B$ satisfies the aforementioned 
property with high probability. 
In other words, for $m=m(q,n)$ sufficiently large, a random linear operator from 
$\ell_2^n$ to $\ell_q^m$ is an approximate isometry. 

To restate this from a geometric perspective, for $m(q,n)$ sufficiently larger than $n$, a random section 
of the unit ball in $\ell_q^m$ is approximately isometric to the unit ball in $\ell_2^n$. In the interest of 
simplicity, we will instead state and use a corollary of the following matrix deviation inequality due to 
Schechtman (see \cite{Schechtman06}, Chapter 11 in \cite{Vershynin17}). 
\begin{theorem}[Schechtman~\cite{Schechtman06}]
\label{matrix:deviation:bound}
    Let $B$ be an $m \times n$ matrix with i.i.d. $\gaussian{0}{1}$ entries. Let $f : \R^m\to \R$ be a 
    positive-homogeneous and subadditive function, and let $b$ be such that 
    $f(y)\leq b \cnorm{2}{y}$ for all $y\in \R^m$. Then for any $T\subset \R^n$, 
    \[
        \sup_{x\in T} |f(Bx)-\Ex{f(Bx)}| = O(b\cdot \gamma(T) + t\cdot \mathrm{rad}(T))
    \]
    with probability at least $1-e^{-t^2}$, 
    where $\mathrm{rad}(T)$ is the radius of $T$, and 
    $\gamma(T)$ is the Gaussian complexity of $T$ defined as 
    \[
        \gamma(T) := \Ex{g\sim \gaussian{0}{I_n}}{\sup_{t\in T} |\mysmalldot{g}{t}|}
    \]
\end{theorem}
The above theorem is established by proving that the random process given by
$X_x := f(Bx) - \Ex{f(Bx)}$  has sub-gaussian increments with respect to $L_2$ and subsequently
appealing to  Talagrand's Comparison tail bound.

We will apply this theorem with $f(\cdot) = \cnorm{q}{\cdot}$, $b=1$ and $T$ being the unit ball 
under $\cnorm{2}{\cdot}$. We first state a known estimate of $\Ex{f(Bx)} = \Ex{\cnorm{q}{Bx}}$ for any fixed $x$ satisfying $\cnorm{2}{x} = 1$. 
Note that when $\cnorm{2}{x} = 1$, $Bx$ has the same distribution as an $m$-dimensional random vector with i.i.d. $\gaussian{0}{1}$ coordinates. 

\begin{theorem}[Biau and Mason~\cite{BM15}]
\label{expected:qnorm}
    Let $X\in\R^m$ be a random vector with i.i.d. $\gaussian{0}{1}$ coordinates. Then for any $q\geq 2$, 
    \[
        \Ex{\cnorm{q}{X}} = m^{1/q}\cdot \gamma_q+O(m^{(1/q) -1)}).
    \]
\end{theorem}

We are now equipped to see that a tall random Gaussian matrix is an approximate isometry (as a linear map 
from $\ell_2^n$ to $\ell_q^m$) with high probability. 
\begin{corollary}
\label{embedding}
    Let $B$ be an $m \times n$ matrix with i.i.d. $\gaussian{0}{1}$ entries where $m = \omega(n^{q/2})$. 
    Then with probability at least $1-e^{-n}$, every vector $x\in \R^n$ 
    satisfies, \[\cnorm{q}{Bx} = (1\pm o(1))\cdot m^{1/q} \cdot \gamma_q  \cdot \cnorm{2}{x}. \]
\end{corollary}

\begin{proof}
    We apply \cref{matrix:deviation:bound} with function $f$ being the $\ell_q$ norm, $b=1$, and $t = \sqrt{n}$. 
    Further we set $T$ to be the $\ell_2$ unit sphere, which yields $\gamma(T) = \Theta(\sqrt{n})$ and 
    $\mathrm{rad}(T) = 1$.  Applying \cref{expected:qnorm} yields that with probability at least 
    $1 - e^{t^2} = 1-e^{-n}$, for all $x$ with $\cnorm{2}{x} = 1$, we have 
\begin{align*}
\left| \cnorm{q}{Bx} - m^{1/q}  \cdot \gamma_q \right| 
& \leq 
\left| \cnorm{q}{Bx} - \Ex{\cnorm{q}{X}} \right| 
+ 
\left| \Ex{\cnorm{q}{X}} - m^{1/q} \cdot  \gamma_q \right|  \\
& \leq O(b \cdot \gamma(T) + t  \cdot \mathrm{rad}(T) + m^{(1/q)-1})  \\
& \leq O(\sqrt{n} + \sqrt{n} + m^{(1/q)-1})\\ 
& \leq o(m^{1/q}). \tag*{\qedhere}
\end{align*}
\end{proof}

We thus obtain the desired constant factor hardness:
\begin{proposition}
\label{hypercontractive:const:factor}
    For any $p>2,~2\leq q< \infty$ and any $\varepsilon>0$, there is no polynomial time algorithm that 
    approximates \onorm{p}{q} (and consequently \onorm{q^*}{p^*}) within a factor of 
    $1/\gamma_{p^*}-\varepsilon$ ~unless $\NP \not\subseteq \BPP$. 
\end{proposition}

\begin{proof}
    By \cref{embedding}, for every $n\times n$ matrix $A$ and a random $m\times n$ matrix $B$ with i.i.d. 
    $\gaussian{0}{1}$ entries ($m = \omega(n^{q/2})$), with probability at least $1-e^{-n}$, we have 
    \[
        \cnorm{p}{q}{BA} = 
        (1\pm o(1))\cdot \gamma_q\cdot m^{1/q}\cdot \cnorm{p}{2}{A}.
    \]
    Thus the reduction $A\rightarrow BA$ combined with \onorm{p}{2} hardness implied by \cref{thm:brs}, 
    yields the claim. 
\end{proof}

The generality of the concentration of measure phenomenon underlying the proof of the 
Dvoretzky-Milman theorem allows us to generalize \cref{hypercontractive:const:factor}, to obtain constant 
factor hardness of maximizing various norms over the $\ell_p$ ball ($p>2$).  In this more general version, 
the strength of our hardness assumption is dependent on the Gaussian width of the dual of the norm being 
maximized. Its proof is identical to that of \cref{hypercontractive:const:factor}. 

\begin{theorem}
\label{p-to-anything:const:factor}
    Consider any $p>2, \varepsilon>0$, and any family $(f_m)_{m\in \N}$ of positive-homogeneous and  
    subadditive functions where $f_m : \R^m\to \R$. Let $(b_m)_{m\in \N}$ be such that $f_m(y) \leq b_m\cdot 
    \cnorm{2}{y}$ for all $y$ and let $N=N(n)$ be such that $\gamma_*(f_N) = \omega(b_N\cdot \sqrt{n})$, 
    where 
    \[
        \gamma^*(f_N) := \Ex{g\sim \gaussian{0}{I_N}}{f_N(g)}.
    \] 
    Then unless $\NP \not\subseteq \BPTIME{N(n)}$, there is no polynomial 
    time $(1/\gamma_{p^*}-\varepsilon)$-approximation algorithm for the problem of computing 
    $\sup_{\norm{p}{x} = 1} f_m(Ax)$, given an $m\times n$ matrix $A$. 
\end{theorem}

\subsection{Derandomized Reduction}\label{sec:derandomization}
In this section,  we show how to derandomize the reduction in \cref{hypercontractive:const:factor}
to obtain  NP-hardness when $q\geq 2$ is an even integer and $p > 2$.  Similarly to~\cref{isometry},
given $A \in \R^{n \times n}$ as a hard instance of \onorm{p}{2}, our strategy is to construct a
matrix $B \in \R^{m \times n}$ and output $BA$ as a hard instance of \onorm{p}{q}.

Instead of requiring $B$ to satisfy $\cnorm{q}{Bx} = (1 \pm o(1)) \cnorm{2}{x}$ for all $x \in
\R^n$, we show that $\cnorm{q}{Bx} \leq (1 + o(1)) \cnorm{2}{x}$ for all $x \in \R^n$  and
$\cnorm{q}{Bx} \geq (1 - o(1)) \cnorm{2}{x}$ when every coordinate of $x$ has the same absolute
value. Since~\cref{thm:brs} ensures that $\cnorm{p}{2}{A}$ is achieved by $x = Ax$ for such a
well-spread $x$ in the completeness case,  $BA$ serves as a hard instance for \onorm{p}{q}. 

\ifconf
We use a construction of $q$-wise independent sets due to Alon, Babai and Itai 
\cite{ABI86} to prove the following derandomized version of our results. The details of the
derandomization can be found in the full version of the paper.
\fi

\ifnotconf
 We use the following construction of $q$-wise independent sets to construct such a $B$ deterministically.

\begin{theorem}[Alon, Babai, and Itai~\cite{ABI86}]
\label{kwise:independence}
    For any $k\in\N$, one can compute a set $S$ of vectors in $\{\pm 1\}^n$ of size $O(n^{k/2})$, in time 
    $n^{O(k)}$, such that the vector random variable $Y$ obtained by sampling uniformly from $S$ satisfies 
    that for any $I\in {[n]\choose k}$, the marginal distribution $Y\restrict{I}$ is the uniform distribution over 
    $\{\pm 1\}^{k}$. 
\end{theorem}

For a matrix $B$ as above, a randomly chosen row behaves similarly to an $n$-dimensional
Rademacher random vector with respect to $\cnorm{q}{\cdot}$. 

\begin{corollary}
\label{derandomized:operator}
    Let $R\in\R^n$ be a vector random variable with i.i.d. Rademacher ($\pm 1$) coordinates. 
    For any even integer $q\geq 2$, there is an $m \times n$ matrix $B$ with $m = O(n^{q/2})$, computable in $n^{O(q)}$ time, 
    such that for all $x\in\R^n$, we have 
    \[
        \cnorm{q}{Bx} = m^{1/q} \cdot \Ex{R}{\mysmalldot{R}{x}^q}^{1/q}.
    \]
\end{corollary}
\begin{proof}\belowdisplayskip=-12pt
    Let $B$ be a matrix, the set of whose rows is precisely $S$.
    By \cref{kwise:independence}, 
    \begin{align*}
        \cnorm{q}{Bx}^q
        =
        \sum_{Y \in S}{\mysmalldot{Y}{x}^q}
        &= 
        m\cdot \Ex{R}{\mysmalldot{R}{x}^q}
        \mper \tag*{\qedhere}
    \end{align*}
\end{proof}

We use the following two results that will bound $\cnorm{p}{q}{BA}$ for the completeness case and the soundness case respectively.

\begin{theorem}[Stechkin~\cite{Stechkin61}]
\label{stechkin:clt}
    Let $R\in\R^n$ be a vector random variable with i.i.d. Rademacher coordinates. Then for any $q\geq 2$ 
    and any $x\in \R^n$ whose coordinates have the same absolute value, 
    \[
        \Ex{\mysmalldot{R}{x}}^{1/q} = (1-o(1))\cdot \gamma_q \cnorm{2}{x}.
    \]
\end{theorem}

\begin{theorem}[Khintchine inequality~\cite{Haagerup81}]
\label{khintchine:clt}
    Let $R\in\R^n$ be a vector random variable with i.i.d. Rademacher coordinates. Then for any $q\geq 2$ 
    and any $x \in \R^n$, 
    \[
        \Ex{\mysmalldot{R}{x}^q}^{1/q} \leq \gamma_q \cdot \cnorm{2}{x}.
    \]
\end{theorem}

We finally prove the derandomimzed version of 
\cref{hypercontractive:const:factor} for even $q \geq 2$. 

\fi

\begin{proposition}
\label{hypercontractive:derandomized:const:factor}
    For any $p > 2, \varepsilon>0$, and any even integer $q\geq 2$, it is NP-hard to approximate 
    \onorm{p}{q} within a factor of ~$1/\gamma_{p^*}  - \varepsilon$.
\end{proposition}

\ifnotconf
\begin{proof}
    Apply \cref{thm:brs} with $r_1 \leftarrow p^*$ and $\varepsilon \leftarrow \varepsilon$.  Given an instance $\phi$ of \threesat, 
    \cref{thm:brs} produces a symmetric matrix $A \in \R^{n \times n}$ in polynomial time as a hard instance of \onorm{p}{2}.  Our instance for \onorm{p}{q} is $BA$ where $B$ is 
    the $m \times n$ matrix given by \cref{derandomized:operator} with $m = O(n^{q/2})$. 
    \begin{itemize}
        \item (Completeness) If $\phi$ is satisfiable, there exists a vector $x\in\{\pm \frac{ 1}{\sqrt{n}}\}^n$ such that $Ax = x$. 
        So we have $\cnorm{q}{BAx} = \cnorm{q}{Bx} = (1-o(1))\cdot m^{1/q} \cdot \gamma_q$, where the last 
        equality uses \cref{derandomized:operator} and \cref{stechkin:clt}. Thus $\cnorm{p}{q}{BA}\geq 
        (1-o(1))\cdot m^{1/q} \cdot \gamma_q$. 
        
        \item (Soundness) If $\phi$ is not satisfiable, then for any $x$ with $\cnorm{p}{x}=1$, 
            \begin{align*}
                &\cnorm{q}{BAx} = m^{1/q} \cdot \Ex{R}{\mysmalldot{R}{Ax}^q}^{1/q} \leq m^{1/q} \cdot \gamma_q\cdot \cnorm{2}{Ax} 
                 \\
                \leq~ & m^{1/q} \cdot  \gamma_q \cdot \cnorm{p}{2}{A} 
                \leq~ m^{1/q} \cdot  \gamma_q \cdot (\gamma_{p^*}-\varepsilon)
            \end{align*}
            where the first inequality is a direct application of~\cref{khintchine:clt}. 
            \hfill\qedhere
   \end{itemize}
\end{proof}
\fi

\subsection{Hypercontractive Norms Productivize}\label{sec:productivize}
We will next amplify our hardness results using the fact that hypercontractive norms
productivize under the  natural operation of Kronecker or tensor product. 
Bhaskara and Vijayraghavan~\cite{BV11} showed this for the special case of $p=q$ and the Harrow  and 
Montanaro~\cite{HM13} showed this for \onorm{2}{4} (via parallel repetition for $\mathrm{QMA(2)}$). 
In this section we prove this claim whenever $p\leq q$. 

\begin{theorem}
\label{productivization}
    Let $A$ and $B$ be $m_1\times n_1$ and $m_2\times n_2$ matrices respectively. Then for any $1\leq p\leq q < \infty$, 
    $\cnorm{p}{q}{A\otimes B} \leq \cnorm{p}{q}{A}\cdot \cnorm{p}{q}{B}$. 
\end{theorem}

\begin{proof}
    We will begin with some notation. Let $a_{i},b_{j}$ respectively denote the $i$-th and $j$-th rows of $A$ 
    and $B$. Consider any $z\in \R^{[n_1]\times [n_2]}$ satisfying $\cnorm{p}{z}=1$. 
    For $k\in [n_1]$, let $z_k\in \R^{n_2}$ denote the vector given by $z_k(\ell) := z(k,\ell)$. 
    For $j\in [m_2]$, let $\barz_j\in \R^{n_1}$ denote the vector given by $\barz_j(k) := 
    \mysmalldot{b_{j}}{z_{k}}$. 
    Finally, for $k\in [n_1]$, let $\lambda_k := \cnorm{p}{z_k}^{p}$ and let $v_k\in \R^{m_2}$ be the vector given 
    by $v_k(j) := |\barz_j(k)|^{p}/\lambda_k$. 
    
    We begin by 'peeling off' $A$:
    \begin{align*}
        \cnorm{q}{(A\otimes B)z}^{q} 
        ~=~ 
        \sum_{i,j} |\mysmalldot{a_i\otimes b_j}{z}|^{q}  
        &~=~
        \sum_{j} \sum_{i} |\mysmalldot{a_i}{\barz_j}|^{q}  \\
        &~=~
        \sum_{j} \cnorm{q}{A\barz_j}^{q}  \\ 
        &~\leq~
        \cnorm{p}{q}{A}^{q}\cdot \sum_{j} \cnorm{p}{\barz_j}^{q}  \\
        &~=~
        \cnorm{p}{q}{A}^{q}\cdot \sum_{j} \inparen{\cnorm{p}{\barz_j}^{p}}^{q/p} 
    \end{align*}
    In the special case of $p=q$, the proof ends here since the expression is a sum of terms of the 
    form $\cnorm{p}{By}^p$ and can thus be upper bounded term-wise by 
    $\cnorm{p}{p}{B}^p\cdot \cnorm{p}{z_k}^p$ which sums to $\cnorm{q}{p}{B}^p$. To handle the case of 
    $q>p$, we will use a convexity argument: 
    \begin{align*}
        &\cnorm{p}{q}{A}^{q}\cdot \sum_{j} \inparen{\cnorm{p}{\barz_j}^{p}}^{q/p}  \\
        ~=~
        &\cnorm{p}{q}{A}^{q}\cdot \sum_{j} \inparen{\sum_k |\barz_j(k)|^{p}}^{q/p}  \\
        ~=~
        &\cnorm{p}{q}{A}^{q}\cdot \cnorm{q/p}{\sum_k \lambda_k\cdot v_k}^{q/p}  
&&(|\barz_j(k)|^{p} = \lambda_k v_k(j)) \\
        ~\leq~
        &\cnorm{p}{q}{A}^{q}\cdot \sum_k \lambda_k \cdot \cnorm{q/p}{v_k}^{q/p}  
        &&(\text{by convexity of  }\norm{q/p}{\cdot}^{q/p} \text{ when } q\geq p) \\
        ~\leq~
        &\cnorm{p}{q}{A}^{q}\cdot \max_k \cnorm{q/p}{v_k}^{q/p}  
    \end{align*}
    
    It remains to show that $\cnorm{q/p}{v_k}^{q/p}$ is precisely $\cnorm{q}{Bz_k}^q/\cnorm{p}{z_k}^q$. 
    \begin{align*}
        \cnorm{p}{q}{A}^{q}\cdot \max_k \cnorm{q/p}{v_k}^{q/p} 
        ~=~
        &\cnorm{p}{q}{A}^{q}\cdot \max_k \frac{1}{\cnorm{p}{z_k}^q}\cdot \sum_j |\barz_j(k)|^q  \\
        ~=~
        &\cnorm{p}{q}{A}^{q}\cdot \max_k \frac{1}{\cnorm{p}{z_k}^q}\cdot \sum_j |\mysmalldot{b_j}{z_k}|^q  \\
        ~=~
        &\cnorm{p}{q}{A}^{q}\cdot \max_k \frac{\cnorm{q}{Bz_k}^q}{\cnorm{p}{z_k}^q}  \\
        ~\leq~
        &\cnorm{p}{q}{A}^{q}\cdot \cnorm{p}{q}{B}^q   
    \end{align*}
    Thus we have established $\cnorm{p}{q}{A\otimes B} \leq \cnorm{p}{q}{A}\cdot \cnorm{p}{q}{B}$. 
    Lastly, the claim follows by observing that the statement is equivalent to the statement obtained by 
    replacing the counting norms with expectation norms. 
\end{proof}

We finally establish super constant NP-Hardness of approximating \onorm{p}{q}, proving \cref{thm:main_hype}.
\begin{proofof}{\cref{thm:main_hype}}
Fix $2 < p \leq q < \infty$. 
\cref{hypercontractive:const:factor} states that there exists $c = c(p, q) > 1$ such that 
any polynomial time algorithm approximating the \onorm{p}{q} of an $n \times n$-matrix $A$ within a factor of $c$ will imply $\NP \subseteq \BPP$. 
Using \cref{productivization}, for any integer $k \in \N$ and $N = n^k$, any polynomial time algorithm approximating the \onorm{p}{q} of an $N \times N$-matrix $A^{\otimes k}$ within a factor of $c^k$ 
implies that $\NP$ admits a randomized algorithm running in time $\poly(N) = n^{O(k)}$. 
Under $\NP \not\subseteq \BPP$, any constant factor approximation algorithm is ruled out by setting $k$ to be a sufficiently large constant. 
For any $\varepsilon > 0$, setting $k = \log^{1/\varepsilon} n$ rules out an approximation factor of 
$c^k = 2^{O(\log^{1 - \varepsilon} N)}$ unless $\NP \subseteq \BPTIME{2^{\log^{O(1)} n}}$.

By duality, the same statements hold for $1 < p \leq q < 2$. 
When $2 < p \leq q$ and $q$ is an even integer, all reductions become deterministic due to \cref{hypercontractive:derandomized:const:factor}.
\end{proofof}


\subsection{A Simple Proof of Hardness for the Case $2 \notin [q,p]$}
\label{sec:reverse}
In this section, we show how to prove an almost-polynomial factor hardness for approximating
\onorm{p}{q} in the non-hypercontractive case when $2 > p \geq q$ (and the case $p \geq q > 2$
follows by duality). 
This result is already known from the work of Bhaskara and Vijayaraghavan \cite{BV11}. We
show how to obtain a more modular proof, composing our previous results with a simple embedding
argument. 
However, while the reduction in \cite{BV11} was deterministic, we will only give a randomized
reduction below.

As in \cite{BV11}, we start with a strong hardness for the \onorm{p}{p}, obtained in
\cref{thm:main_hype}. While the reduction in \cite{BV11} relied on special properties of the
instance for \onorm{\ell_p}{\ell_p}, we can simply use the following embedding result of Schechtman
\cite{Schechtman87} (phrased in a way convenient for our application). 
\begin{theorem}[Schechtman \cite{Schechtman87}, Theorem 5]
Let $q<p<2$ and $\eps>0$. Then, there exists a polynomial time samplable distribution $\calD$ 
on random matrices in $\R^{m \times
  n}$ with $m = \Omega_{\eps}(n^{3})$,  such that with probability $1-o(1)$, we have for every
$x \in R^n$,  $\norm{\ell_q}{Bx} ~=~ (1 \pm \eps) \cdot \norm{\ell_p}{x}$.
\label{thm:schechtman}
\end{theorem}
In fact the distribution $\calD$ is based on $p$-stable distributions. 
While the theorem in \cite{Schechtman87} does not mention the high probability bound or
samplability, it is easy to modify the proof to obtain there properties. We provide a proof sketch
below for completeness. We note that Schechtman obtains a stronger bound of $O(n^{1+p/q})$ on the
dimension $m$ of the $\ell_q$ space, which requires a more sophisticated argument using ``Lewis
weights''. However, we only state weaker $O(n^3)$ bound above, which suffices for our purposes and
is easier to convert to a samplable distribution. 

We first prove the following hardness result for approximating \onorm{p}{q} in the
reverse-hypercontractive case, using \cref{thm:schechtman}.
\begin{theorem}
For any $p, q$ such that $1 < q \leq p < 2$ or $2 < q \leq p < \infty$ and $\epsilon > 0$,  there
is no polynomial time algorithm that approximates the \onorm{p}{q} of an $n \times n$ matrix within
a factor $2^{\log^{1 - \epsilon} n}$ unless $\NP \subseteq \BPTIME{2^{(\log n)^{O(1)}}}$. 
\label{thm:reverse}
\end{theorem}
\begin{proof}
We consider the case $1 < q \leq p < 2$ (the other case follows via duality).
\cref{thm:main_hype} gives a reduction from SAT on $n$ variables, approximating the \onorm{p}{p} of
matrices $A \in \R^{N \times N}$ with $N = 2^{(\log n)^{O(1/\eps)}}$, within a factor $2^{(\log
  N)^{1-\eps}}$. Sampling a matrix $B$ from the distribution $\calD$ given by \cref{thm:schechtman}
(with dimension $N$) gives that it is also hard to approximate $\norm{p}{q}{BA} \approx \norm{p}{p}{A}$,
within a factor $2^{(\log N)^{1-\eps}}$.
\end{proof}
We now give a sketch of the proof of \cref{thm:schechtman} including the samplability condition. The
key idea is to embed the space $\ell_p^n$ into the infinite-dimensional space $L_q$ (for $0 \leq q
\leq p < 2$) using $p$-stable random variables. The corresponding subspace of $L_q$ can then be
embedded into $\ell_q^m$ if the random variables (elements of $L_q$) constructed in the previous
space are bounded in $L_{\infty}$ norm. This is the content of the following claim.
\begin{claim}[Schechtman \cite{Schechtman87}, Proposition 4]
\label{clm:schechtman}
Let $\eps > 0$ and $\Omega$ be an efficiently samplable probability space and let $V$ be an
$n$-dimensional subspace of $L_q(\Omega)$, such that
\[
M ~\defeq~ \sup\inbraces{\norm{L_{\infty}}{f} ~\mid~ \norm{L_{q}}{f} \leq 1, f \in V} ~<~ \infty \mper
\]
Then there exists a polynomial time samplable distribution $\calD$ over linear operators $T:
L_q(\Omega) \to \R^{m}$ for $m = C(\eps, q) \cdot n \cdot M^q$ such that with
probability $1 - o(1)$, we have that for every $f \in V$, $\norm{\ell_q}{Tf} = (1\pm \eps) \cdot \norm{L_q}{f}$.
\end{claim}
\begin{proofsketch}
The linear operator is simply defined by sampling $x_1, \ldots, x_m \sim \Omega$ independently, and
taking
\[
Tf ~\defeq~ \frac{1}{m^{1/q}} \cdot \inparen{f(x_1), \ldots, f(x_m)} \qquad \forall f\mper
\]
The proof then follows by concentration bounds for $L_{\infty}$-bounded random variables, and a
union bound over an epsilon net for the space $V$.
\end{proofsketch}
The problem then reduces to constructing an embedding of $\ell_p^n$ into $L_q$, which is bounded in
$L_{\infty}$ norm. While a simple embedding can be constructed using $p$-stable distributions,
Schechtman uses a clever reweighting argument to control the $L_{\infty}$ norm. We show below that a
simple truncation argument can also be used to obtain a somewhat crude bound on the $L_{\infty}$
norm, which suffices for our purposes and yields an easily samplable distribution. 

We collect below the relevant facts about $p$-stable random variables needed for our argument, which
can be found in many well-known references, including \cite{Indyk06, AlbiacK06}. 
\begin{fact}
For all $p \in (0,2)$, there exist (normalized) $p$-stable random variables $Z$ satisfying the
following properties:
\begin{enumerate}
\item For $Z_1, \ldots, Z_n$ iid copies of $Z$, and for all $a \in \R^n$, the random variable
\[
S ~\defeq~ \frac{a_1 \cdot Z_1 + \cdots + a_n \cdot Z_n}{\norm{\ell_p}{a}} \mcom
\]
has distribution identical to $Z$.
\item For all $q < p$, we have 
\[
C_{p,q} ~\defeq~ \norm{L_q}{Z} ~=~ \inparen{\Ex{\abs{Z}^q}}^{1/q} ~<~ \infty \mper
\]
\item There exists a constant $C_p$ such that for all $t > 0$,
\[
\Pr{\abs{Z} \geq t} ~<~ \frac{C_p}{t} \mper
\]
\item $Z$ can be sampled by choosing $\theta \in_R [-\pi/2, \pi/2]$, $r \in_R [0,1]$, and taking
\[
Z ~=~ \frac{\sin(p\theta)}{(\cos(\theta))^{1/p}} \cdot \inparen{\frac{\cos((1-p) \cdot
    \theta)}{\ln(1/r)}}^{(1-p)/p} \mper
\]
\end{enumerate}
\end{fact}
We now define an embedding of $\ell_p^n$ into $L_q$ with bounded $L_{\infty}$, using truncated
$p$-stable random variables. Let $Z = (Z_1, \ldots, Z_n)$ be a vector of 
iid $p$-stable random variables as above, and let $B$ be a parameter to be chosen later. We consider
the random variables
\[
\Delta(Z) ~\defeq~ \indicator{\exists i \in [n] ~\abs{Z_i} > B} 
\quad \text{and} \quad
Y ~\defeq~ (1 - \Delta(Z)) \cdot Z ~=~ \indicator{\forall i \in [n] ~\abs{Z_i} \leq B} \cdot Z\mper
\]
For all $a \in \R^n$, we define the (linear) embedding
\[
\phi(a) ~\defeq~ \frac{ \ip{a,Y}}{C_{p,q}} ~=~ \frac{\ip{a,Z}}{C_{p,q}} - \Delta(Z) \cdot
\frac{\ip{a,Z}}{C_{p,q}}\mper 
\] 
By the properties of $p$-stable distributions, we know that $\norm{L_q}{\ip{a,Z}/C_{p,q}} =
\norm{\ell_p}{a}$ for all $a \in \R^n$. By the following claim, we can choose $B$ so that the second
term only introduces a small error.
\begin{claim}
For all $\eps > 0$, there exists $B = O_{p,q,\eps}(n^{1/p})$ such that for the embedding $\phi$ defined
above
\[
\abs{\norm{L_q}{\phi(a)} - \norm{\ell_p}{a}} ~\leq~ \eps \cdot \norm{\ell_p}{a} \mper
\]
\end{claim}
\begin{proof}
By triangle inequality, it suffices to bound $\norm{L_q}{\Delta(Z) \cdot \ip{a,Z}}$ by $\eps \cdot
C_{p,q} \cdot \norm{\ell_p}{a}$. Let $\delta > 0$ be such that $(1+\delta) \cdot q < p$. Using the
fact that $\Delta(Z)$ is Boolean and \Holder's inequality, we observe that
\begin{align*}
\norm{L_q}{\Delta(Z) \cdot \ip{a,Z}}
&~=~
\inparen{\Ex{\abs{\ip{a,Z}}^q \cdot \Delta(Z)}}^{1/q} \\
&~\leq~
\inparen{\Ex{\abs{\ip{a,Z}}^{q(1+\delta)}}}^{1/(q(1+\delta))} \cdot
  \inparen{\Ex{\Delta(Z)}}^{\delta/(q(1+\delta))} \\
&~=~
C_{p,(1+\delta)q} \cdot \norm{\ell_p}{a} \cdot \inparen{\Pr{\exists i \in [n] ~\abs{Z_i} \geq
  B}}^{\delta/(q(1+\delta))} \\
&~\leq~
C_{p,(1+\delta)q} \cdot \norm{\ell_p}{a} \cdot \inparen{n \cdot \frac{C_p}{B^p}}^{\delta/(q(1+\delta))}
\end{align*}
Thus, choosing $B = O_{\eps, p, q}(n^{1/p})$ such that
\[
\frac{C_{p,(1+\delta)q}}{C_{p,q}} \cdot \inparen{n \cdot \frac{C_p}{B^p}}^{\delta/(q(1+\delta))}
~\leq~ \eps
\]
proves the claim.
\end{proof}
Using the value of $B$ as above, we now observe a bound on $\norm{L_{\infty}}{\phi(a)}$.
\begin{claim}
Let $B = O_{\eps,p,q}(n^{1/p})$ be chosen as above. Then, we have that
\[
M ~\defeq~ \sup\inbraces{\norm{L_{\infty}}{\ip{a,Y}} ~\mid~ \norm{L_{q}}{\ip{a,Y}} \leq 1} ~=~
O_{\eps,p,q}(n) \mper 
\]
\end{claim}
\begin{proof}
By the choice of $B$, we have that $\norm{L_{q}}{\ip{a,Y}} \geq (1-\eps) \norm{\ell_p}{a}$. Thus, we
can assume that $\norm{\ell_p}{a} \leq 2$. \Holder's inequality then gives for all such $a$,
\begin{align*}
\abs{\ip{a,Y}} 
&~\leq~ 
\norm{\ell_1}{a} \cdot \norm{\ell_{\infty}}{Y} \\
&~\leq~
n^{1-1/p} \cdot \norm{\ell_p}{a} \cdot B \\
&~\leq~ 2 \cdot n^{1-1/p} \cdot B ~=~ O_{\eps, p, q}(n) \mcom
\end{align*}
which proves the claim.
\end{proof}
Using the above bound on $M$ in \cref{clm:schechtman} gives a bound of $m = O_{\eps,p,q}(n^{q+1}) =
O_{\eps,p,q}(n^3)$. Moreover, the distribution over embeddings is efficiently samplable, since it
obtained by truncating $p$-stable random variables. This completes the proof  of \cref{thm:schechtman}.



\bibliographystyle{alpha}
\bibliography{p-to-q-refs}

\newcommand{\etalchar}[1]{$^{#1}$}
\begin{thebibliography}{AMMN06}

\bibitem[ABH{\etalchar{+}}05]{ABHKS05}
Sanjeev Arora, Eli Berger, Elad Hazan, Guy Kindler, and Muli Safra.
\newblock On non-approximability for quadratic programs.
\newblock In {\em Foundations of Computer Science, 2005. FOCS 2005. 46th Annual
  IEEE Symposium on}, pages 206--215. IEEE, 2005.

\bibitem[ABI86]{ABI86}
Noga Alon, L{\'a}szl{\'o} Babai, and Alon Itai.
\newblock A fast and simple randomized parallel algorithm for the maximal
  independent set problem.
\newblock {\em Journal of algorithms}, 7(4):567--583, 1986.

\bibitem[AK06]{AlbiacK06}
Fernando Albiac and Nigel~John Kalton.
\newblock {\em Topics in {Banach} space theory}, volume 233.
\newblock Springer, 2006.

\bibitem[AMMN06]{AMMN06}
Noga Alon, Konstantin Makarychev, Yury Makarychev, and Assaf Naor.
\newblock Quadratic forms on graphs.
\newblock {\em Inventiones mathematicae}, 163(3):499--522, 2006.
\newblock Conference version in STOC '05.

\bibitem[BBH{\etalchar{+}}12]{BBHKSZ12}
Boaz Barak, Fernando~GSL Brandao, Aram~W Harrow, Jonathan Kelner, David
  Steurer, and Yuan Zhou.
\newblock Hypercontractivity, sum-of-squares proofs, and their applications.
\newblock In {\em Proceedings of the forty-fourth annual ACM symposium on
  Theory of computing}, pages 307--326. ACM, 2012.

\bibitem[BGG{\etalchar{+}}18]{BGGLT18b}
V.~{Bhattiprolu}, M.~{Ghosh}, V.~{Guruswami}, E.~{Lee}, and M.~{Tulsiani}.
\newblock Approximating operator norms via generalized {K}rivine rounding.
\newblock 2018.
\newblock Manuscript.

\bibitem[BH15]{BH15}
Fernando~GSL Brandao and Aram~W Harrow.
\newblock Estimating operator norms using covering nets.
\newblock {\em arXiv preprint arXiv:1509.05065}, 2015.

\bibitem[Bis11]{Biswal11}
Punyashloka Biswal.
\newblock Hypercontractivity and its applications.
\newblock {\em arXiv preprint arXiv:1101.2913}, 2011.

\bibitem[BM15]{BM15}
G{\'e}rard Biau and David~M Mason.
\newblock High-dimensional p-norms.
\newblock In {\em Mathematical Statistics and Limit Theorems}, pages 21--40.
  Springer, 2015.

\bibitem[BRS15]{BRS15}
Jop Bri{\"e}t, Oded Regev, and Rishi Saket.
\newblock Tight hardness of the non-commutative {G}rothendieck problem.
\newblock In {\em Foundations of Computer Science (FOCS), 2015 IEEE 56th Annual
  Symposium on}, pages 1108--1122. IEEE, 2015.

\bibitem[BV11]{BV11}
Aditya Bhaskara and Aravindan Vijayaraghavan.
\newblock Approximating matrix p-norms.
\newblock In {\em Proceedings of the twenty-second annual ACM-SIAM symposium on
  Discrete Algorithms}, pages 497--511. SIAM, 2011.

\bibitem[Gro56]{Grothendieck56}
Alexandre Grothendieck.
\newblock {\em R{\'e}sum{\'e} de la th{\'e}orie m{\'e}trique des produits
  tensoriels topologiques}.
\newblock Soc. de Matem{\'a}tica de S{\~a}o Paulo, 1956.

\bibitem[GRSW16]{GRSW16}
Venkatesan Guruswami, Prasad Raghavendra, Rishi Saket, and Yi~Wu.
\newblock Bypassing {UGC} from some optimal geometric inapproximability
  results.
\newblock {\em ACM Transactions on Algorithms (TALG)}, 12(1):6, 2016.
\newblock Conference version in SODA '12.

\bibitem[Haa81]{Haagerup81}
Uffe Haagerup.
\newblock The best constants in the khintchine inequality.
\newblock {\em Studia Mathematica}, 70(3):231--283, 1981.

\bibitem[HM13]{HM13}
Aram~W Harrow and Ashley Montanaro.
\newblock Testing product states, quantum {M}erlin-{A}rthur games and tensor
  optimization.
\newblock {\em Journal of the ACM (JACM)}, 60(1):3, 2013.

\bibitem[HNW16]{HNW16}
Aram~W Harrow, Anand Natarajan, and Xiaodi Wu.
\newblock Limitations of semidefinite programs for separable states and
  entangled games.
\newblock {\em arXiv preprint arXiv:1612.09306}, 2016.

\bibitem[Ind06]{Indyk06}
Piotr Indyk.
\newblock Stable distributions, pseudorandom generators, embeddings, and data
  stream computation.
\newblock {\em Journal of the ACM (JACM)}, 53(3):307--323, 2006.

\bibitem[Kho02]{Khot02}
Subhash Khot.
\newblock Hardness results for coloring 3-colorable 3-uniform hypergraphs.
\newblock In {\em Foundations of Computer Science, 2002. Proceedings. The 43rd
  Annual IEEE Symposium on}, pages 23--32. IEEE, 2002.

\bibitem[KMW18]{KMW18}
Aditya Krishnan, Sidhanth Mohanty, and David~P. Woodruff.
\newblock On sketching $q$ to $p$ norms.
\newblock Manuscript, 2018.

\bibitem[KN08]{KN08}
Subhash Khot and Assaf Naor.
\newblock Linear equations modulo 2 and the l\_1 diameter of convex bodies.
\newblock {\em SIAM Journal on Computing}, 38(4):1448--1463, 2008.

\bibitem[KN12]{KN12}
Subhash Khot and Assaf Naor.
\newblock Grothendieck-type inequalities in combinatorial optimization.
\newblock {\em Communications on Pure and Applied Mathematics},
  65(7):992--1035, 2012.

\bibitem[KNS10]{KNS10}
Guy Kindler, Assaf Naor, and Gideon Schechtman.
\newblock The {UGC} hardness threshold of the {L}p {G}rothendieck problem.
\newblock {\em Mathematics of Operations Research}, 35(2):267--283, 2010.
\newblock Conference version in SODA '08.

\bibitem[KO09]{KO09}
Subhash Khot and Ryan O'Donnell.
\newblock {SDP} gaps and {UGC}-hardness for {Max-Cut-Gain}.
\newblock {\em Theory OF Computing}, 5:83--117, 2009.

\bibitem[Kri77]{Krivine77}
Jean-Louis Krivine.
\newblock Sur la constante de {G}rothendieck.
\newblock {\em CR Acad. Sci. Paris Ser. AB}, 284(8):A445--A446, 1977.

\bibitem[Nes98]{Nesterov98}
Yurii Nesterov.
\newblock Semidefinite relaxation and nonconvex quadratic optimization.
\newblock {\em Optimization methods and software}, 9(1-3):141--160, 1998.

\bibitem[Pis12]{Pisier12}
Gilles Pisier.
\newblock Grothendieck's theorem, past and present.
\newblock {\em Bulletin of the American Mathematical Society}, 49(2):237--323,
  2012.

\bibitem[Ree91]{Reeds91}
JA~Reeds.
\newblock A new lower bound on the real {G}rothendieck constant.
\newblock Manuscript, 1991.

\bibitem[RS09]{RS09}
Prasad Raghavendra and David Steurer.
\newblock Towards computing the {G}rothendieck constant.
\newblock In {\em Proceedings of the Twentieth Annual ACM-SIAM Symposium on
  Discrete Algorithms}, pages 525--534. Society for Industrial and Applied
  Mathematics, 2009.

\bibitem[RS10]{RS10}
Prasad Raghavendra and David Steurer.
\newblock Graph expansion and the unique games conjecture.
\newblock In {\em Proceedings of the 42nd annual ACM Symposium on Theory of
  Computing}, pages 755--764. ACM, 2010.

\bibitem[Sch87]{Schechtman87}
Gideon Schechtman.
\newblock More on embedding subspaces of$\backslash$(l\_p$\backslash$)
  in$\backslash$(l\^{} n\_r$\backslash$).
\newblock {\em Compositio Math}, 61(2):159--169, 1987.

\bibitem[Sch06]{Schechtman06}
Gideon Schechtman.
\newblock Two observations regarding embedding subsets of euclidean spaces in
  normed spaces.
\newblock {\em Advances in Mathematics}, 200(1):125--135, 2006.

\bibitem[Ste61]{Stechkin61}
Sergei~Borisovich Stechkin.
\newblock On best lacunary systems of functions.
\newblock {\em Izvestiya Rossiiskoi Akademii Nauk. Seriya Matematicheskaya},
  25(3):357--366, 1961.

\bibitem[Ste05]{Steinberg05}
Daureen Steinberg.
\newblock Computation of matrix norms with applications to robust optimization.
\newblock {\em Research thesis, Technion-Israel University of Technology},
  2005.

\bibitem[vB72]{Beek72}
Paul van Beek.
\newblock An application of {F}ourier methods to the problem of sharpening the
  {B}erry-{E}sseen inequality.
\newblock {\em Zeitschrift f{\"u}r Wahrscheinlichkeitstheorie und verwandte
  Gebiete}, 23(3):187--196, 1972.

\bibitem[Ver17]{Vershynin17}
Roman Vershynin.
\newblock {\em High Dimensional Probability}.
\newblock 2017.

\end{thebibliography}

\ifnotconf
\appendix

\section{Dictatorship Test}\label{sec:dic}
\newcommand{\innernorm}{\cnorm}
\newcommand{\outernorf}{\enorm}

First we prove an implication of Berry-\Esseen estimate for fractional moments
(similar to Lemma 3.3 of \cite{GRSW16}, see also \cite{KNS10}).
\begin{lemma}
There exist universal constants $c>0$ and $\delta_0>0$ such that the following statement is true. 
    If $X_1,\cdots,X_n$ are bounded independent random variables with $\abs{X_i}\le 1$,
    $\Ex{X_i}=0$ for $i\in[n]$, and $\sum_{i\in[n]}\Ex{X_i^2}=1$,
    $\sum_{i\in[n]}\Ex{\abs{X_i}^3}\leq \delta$ for some $0 < \delta < \delta_0$, then
    for every $p\geq 1$:
    \[
        \inparen{\Ex{\abs{\sum_{j=1}^n X_j}^p}}^{\frac{1}{p}} \leq \gamma_p\cdot
        \inparen{1+c\delta\inparen{\log\inparen{\nfrac{1}{\delta}}}^{\frac{p}{2}}}.
    \]
    \label{lem:clt-analytic-sparse}
\end{lemma}

Now we state and prove the main lemma of this section:
\begin{lemma}\label{lem:dict}
    Let $f : \{ \pm 1 \}^R \to \R$ be a linear 
    function for
    some positive integer $R \in \N$ and $\hatf : [R] \to \R$ be its linear Fourier
    coefficients defined by 
    \[ \hatf(i) \defeq \Ex{x \in \{ \pm 1 \}^R}{x_i f(x)}.\]
    For all $\varepsilon > 0$, there exists $\delta > 0$
    such that if $\enorm{r}{f} > (\gamma_r + \varepsilon) \cnorm{2}{\hatf}$ then
    $\cnorm{4}{\hatf} > \delta \cnorm{2}{\hatf}$ for all $1 \leq r < 2$.
\end{lemma}
\begin{proof}
    We will prove this lemma by the method of contradiction.
    Let us assume $\cnorm{4}{\hatf} \leq \delta\cnorm{2}{\hatf}$, for $\delta$ to be
    fixed later.

    Let us define $y_i\defeq \frac{\hatf(i)}{\cnorm{2}{\hatf}}$.
    Then, for all $x\in\on^R$,  
    \[
        g(x)\defeq \sum_{i\in[n]} x_i\cdot y_i = \frac{f(x)}{\cnorm{2}{\hatf}}\mper
    \]
    Let $Z_i=x_i\cdot y_i$ be the random variable when $x_i$ is independently uniformly
    randomly chosen from $\on$. Now
    \[
        \sum_{i\in[n]}\Ex{Z_i^2}=\sum_{i\in[n]}\frac{\hatf(i)^2}{\cnorm{2}{\hatf}^2}=1\mper
    \]
    and
    \[
        \sum_{i\in[n]}\Ex{\abs{Z_i}^3} = \sum_{i\in[n]}
        \frac{\abs{\hatf(i)}^3}{\cnorm{2}{\hatf}^3}= \sum_{i\in[n]}
        \frac{\abs{\hatf(i)}^2}{\cnorm{2}{\hatf}^2}\cdot \frac{\abs{\hatf(i)}}{\cnorm{2}{\hatf}}
        \leq \frac{\cnorm{4}{\hatf}^2}{\cnorm{2}{\hatf}^2} \leq \delta^2 \mcom
    \]
    where the penultimate inequality follows from Cauchy-Schwarz ineqality.

    Hence, by applying \cref{lem:clt-analytic-sparse}
    on the random variables $Z_1,\cdots, Z_n$, we get:
    \begin{align*}
        \frac{\enorm{r}{f}}{\cnorm{2}{\hatf}}=
        \enorm{r}{g} & = \inparen{ \Ex{x\in \on^n} {\abs{g(x)}^r}}^\frac{1}{r}\\
        & = \inparen{ \Ex{x\in \on^n} {\abs{\sum_{i\in[n]}Z_i}^r}}^\frac{1}{r}\\
        & \leq \gamma_r\inparen{1+c\delta^2 \inparen{\log{\frac{1}{\delta}}}^{r}}\\
    \end{align*}
    We choose $\delta>0$ small enough (since $1 \leq r < 2$, setting $\delta<
    \frac{\sqrt{\varepsilon}}{\min(\delta_0,\sqrt{\gamma_2}\log\frac{c\gamma_2}{\varepsilon})}
= \frac{\sqrt{\varepsilon}}{\min(\delta_0, \log\frac{c}{\varepsilon})}$
    suffices) so that $\delta^2(\log\frac{1}{\delta})^r <
    \frac{\varepsilon}{c\gamma_r}$. For this choise of $\delta$, we get:
    $\enorm{r}{f}\leq (\gamma_r+\varepsilon)\cnorm{2}{\hatf}$ -- a contradiction.
    And hence the proof follows.
\end{proof}

Finally we prove \cref{lem:clt-analytic-sparse}:
\begin{proofof}{\cref{lem:clt-analytic-sparse}}
    The proof is almost similar to that of Lemma 2.1 of \cite{KNS10}.
    From Berry-\Esseen theorem (see \cite{Beek72} for the constant), we get that:
    \[
        \Pr{\abs{\sum_{i=1}^nX_i}\geq u} \leq \Pr{\abs{g}\geq u} +
        2\sum_{i=1}^n\Ex{\abs{X_i}^3} \leq \Pr{\abs{g}\geq u} + 2\delta\mcom
    \] for every $u>0$ and where $g \sim \gaussian{0}{1}$.
    By Hoeffding's lemma, 
\[
\Pr{\abs{\sum_{i\in[n]} X_i}\ge t} < 2 \ee^{-2t^2}
\]
    for every $t>0$.
    Combining the above observations, we get:
    \begin{align*}
        \Ex{\abs{\sum_{i=1}^n X_i}^p} & = \int_{0}^{\infty} p u^{p-1} 
            \Pr{\abs{\sum_{i=1}^nX_i}\geq u}du \\ 
        & \leq \int_{0}^{a} pu^{p-1}\Pr{\abs{g}>u} du + 2\delta a^p +
            2\int_{a}^{\infty} pu^{p-1} \ee^{-2u^2} du\\
        & = \sqrt{\frac{2}{\pi}}\int_{0}^{a} u^p \ee^{-\nfrac{u^2}{2}}du + 2\delta
        a^p + \frac{2p}{2^{\frac{p-1}{2}}}\int_{2a^2}^{\infty} z^{\frac{p+1}{2}-1} \ee^{-z} dz\\
        & =\gamma_p^p -\sqrt{\frac{2}{\pi}}\int_{a}^{\infty} u^p\ee^{-\nfrac{u^2}{2}}du + 2\delta a^p
        +\Gamma\inparen{\frac{p+1}{2},2a^2} ~ \mcom
    \end{align*}
    where $\Gamma(\cdot,\cdot)$ is the upper incomplete gamma function and $a$ is a large constant determined later depending on $\delta$ and $p$.
	The second term is bounded as 
    \begin{align*}
        \int_{a}^{\infty}u^p\ee^{-\nfrac{u^2}{2}}du = a^{p-1}
        \ee^{-\nfrac{a^2}{2}} + (p-1) \int_{a}^{\infty}
        u^{p-2}\ee^{-\nfrac{u^2}{2}} du \leq a^{p-1}
        \ee^{-\nfrac{a^2}{2}} + \frac{p-1}{a^2} \int_{a}^{\infty}
        u^{p}\ee^{-\nfrac{u^2}{2}} du \mper
    \end{align*}
    Hence $\int_{a}^{\infty}u^p\ee^{-\nfrac{u^2}{2}}du \leq
    \frac{a^{p+1}e^{-\nfrac{a^2}{2}}}{1+a^2-p}$.
    
    We know, $\Gamma(\nfrac{p+1}{2},x)\to x^{\frac{p-1}{2}}\ee^{-x}$ as $x\to\infty$.
    We choose $a=\gamma_p \sqrt{\log\frac{1}{\delta}}$. Hence there exists $\delta_0$
    so that for all small enough
    $\delta<\delta_0$, we have $\Gamma(\nfrac{p+1}{2},2a^2) \sim 2^{\frac{p-1}{2}}a^{p-1}
    \delta^{2\gamma_p^2}\ll \delta a^p$ where the last inequality follows from 
    the fact that $2\gamma_p^2 > 1$ (as $p>1$).
    Putting all this together, we get:
    \[
        2\delta a^p +\Gamma\inparen{\frac{p+1}{2},2a^2}-
        \sqrt{\frac{2}{\pi}}\int_{a}^{\infty}u^p\ee^{-\nfrac{u^2}{2}}du 
        \ll 3\delta a^p -\sqrt{\frac{2}{\pi}}\frac{a^{p+1}e^{-\nfrac{a^2}{2}}}{1+a^2-p}
        \leq
        c\gamma_p^p \delta \inparen{\log{\frac{1}{\delta}}}^{\nfrac{p}{2}}\mcom
    \]
    where $c$ is an absolute constant independent of $a$ and $p$.
    This completes the proof of the lemma.
\end{proofof}

\fi

\end{document}